\documentclass[12pt]{article}
\usepackage{amssymb}
\usepackage{natbib}
\usepackage{mathrsfs}
\usepackage{amsmath,amssymb}
\usepackage{amsthm}
\usepackage{graphicx,color}
\usepackage{setspace}
\usepackage[top=1in, bottom=1in, left=1in, right=1in]{geometry}
\usepackage{subfig}
\usepackage{soul}
\setstcolor{red}
\usepackage{multirow}
\usepackage{hhline}

\usepackage{algorithm}
\usepackage{algorithmic}   
\usepackage{multirow}

\renewcommand{\baselinestretch}{\vv}

\newtheorem{theorem}{Theorem}

\newtheorem{lemma}{Lemma}
\newtheorem{corollary}{Corollary}[section]

\title{Parallelizing MCMC via Weierstrass Sampler}

\author{Xiangyu Wang and David B. Dunson
}
\date{\today}
\begin{document}
\maketitle
\begin{abstract} 
With the rapidly growing scales of statistical problems, subset based communication-free parallel MCMC methods are a promising future for large scale Bayesian analysis. In this article, we propose a new Weierstrass sampler for parallel MCMC based on independent subsets. The new sampler approximates the full data posterior samples via combining the posterior draws from independent subset MCMC chains, and thus enjoys a higher computational efficiency. We show that the approximation error for the Weierstrass sampler is bounded by some tuning parameters and provide suggestions for choice of the values. Simulation study shows the Weierstrass sampler is very competitive compared to other methods for combining MCMC chains generated for subsets, including averaging and kernel smoothing.

\end{abstract}
{\it Keywords}: Big data; Communication-free; Embarassingly parallel; MCMC; Scalable Bayes; Subset sampling; Weierstrass transform.

\section{Introduction}
The explosion in the collection and interest in big data in recent years has brought new challenges to modern statistics.  Bayesian analysis, which has benefited from the ease and generality of sampling-based inference, is now suffering as a result of the huge computational demand of posterior sampling algorithms, such as Markov chain Monte Carlo (MCMC), in large scale settings.  MCMC faces several bottlenecks in big data problems due to the increasing expense in likelihood calculations, to the need for updating latent variables specific to each sampling unit, and to increasing mixing problems as the posterior becomes more concentrated in large samples.
\par
To accelerate computation, efforts have been focused in three main directions.  The first is to parallelize computation of the likelihood at each MCMC iteration
\citep{Agarwal:Duchi:2012, Smola:Naray:2010}. Under conditional independence, calculations can be conducted separately for mini batches of data stored on different machines, with results fed back to a central processor.  This approach requires communication within each iteration, which limits overall speed.  A second strategy focuses on accelerating expensive gradient calculations in Langevin and Hamiltonian Monte Carlo \citep{Neal:2010} using stochastic approximation based on random mini batches of data \citep{Welling:Teh:2011, Ahn:etal:2012}.  A third direction is motivated by the independent product equation \eqref{eq:IPE} from \cite{Bailer:Jones:2011}. The data are partitioned into mini batches, MCMC is run independently for each batch without communication, and the chains are then combined to mimic draws from the full data posterior distribution.  By running independent MCMC chains, this approach bypasses communication costs until the combining step and increases MCMC mixing rate, as the subset posteriors are based on smaller sample sizes and hence effectively annealed.  The main open question for this approach is how to combine the chains to obtain an accurate approximation? 
\par

To address this question, \cite{Scott:Blocker:Bonassi:2013} suggested to use averaged posterior draws to approximate the true posterior samples. \cite{Neis:etal:2013} and \cite{White:etal:2013} instead make use of kernel density estimation to approximate the subset posterior densities, and then approximate the true posterior density following \eqref{eq:IPE}.  Based on our experiments, these methods have adequate performance only in specialized settings, and exhibit poor performance in many other cases, such as when parameter dimensionality increases and large sample Gaussian approximations are inadequate.  We propose a new combining method based on the independent product equation \eqref{eq:IPE}, which attempts to address these problems.

In Section 2, we first describe problems that arise in using \eqref{eq:IPE} in a naive way, briefly state the motivation for our method, and provide theoretic justification for the approximation error.  In Section 3 we then describe specific algorithms and discuss tuning parameter choice.  Section 4 assesses the method using extensive examples.  Section 5 contains a discussion, and proofs are included in the appendix.

\section{Motivation}
The fundamental idea for the new method is as follows.  For a parametric model $p(X|\theta)$, assume the data contain $n$ conditionally independent observations, which are partitioned into $m$ non-overlapping subsets, $X = \{X_1,\cdot\cdot\cdot, X_m\}$.  The following relationship holds between the posterior distribution for the full data set and posterior distributions for each subset, 
\begin{align*}
  p(\theta|X) &\propto p(X|\theta)p(\theta) = \bigg\{ \prod_{i=1}^m p(X_i|\theta)\bigg\} p(\theta) \propto \bigg\{ \prod_{i=1}^m \frac{p(\theta|X_i)}{p_i(\theta)} \bigg\} p(\theta)\\
& = \bigg\{ \prod_{i=1}^m p(\theta|X_i) \bigg\} \bigg\{ \frac{p(\theta)}{\prod_{i=1}^m p_i(\theta)} \bigg\},
\end{align*}
where $p(\theta)$ is the prior distribution for the full data set and $p_i(\theta)$ is that for subset $i$. As we will only be obtaining draws from the subset posteriors to approximate draws from $p(\theta|X)$, we have flexibility in the choice of $p_i(\theta)$.  If we further require $p(\theta) = \prod_{i=1}^m p_i(\theta)$, the above equation can be reformulated as 
\begin{align}
  p(\theta|X) \propto p(\theta|X_1)p(\theta|X_2)\cdots p(\theta|X_m) =  \prod_{i=1}^m p(\theta|X_i), \label{eq:IPE}
\end{align}
which we refer to as the independent product equation. This equation indicates that under the independence assumption, the posterior density of the full data can be represented by the product of subset posterior densities if the subsets together form a partition of the original set. However, despite this concise relationship, sampling from the product of densities remains a difficult issue. \cite{Scott:Blocker:Bonassi:2013} use a convolution product to approximate this equation, resulting in an averaging method. 
This approximation is adequate when the subset posteriors are close to Gaussian, as is expected to hold in many parametric models for sufficiently large subset sample sizes due to the Bernstein-von Mises theorem \citep{Van:1998}.
\par
Another intuitive way to apply \eqref{eq:IPE} is to use kernel based non-parametric density estimation (kernel smoothing). Using kernel smoothing, one obtains a closed form approximation to each subset posterior, with these approximations multiplied together to approximate the full data posterior.  This idea has been implemented recently by 
 \cite{Neis:etal:2013}, but suffers from several drawbacks. 
\begin{enumerate}
\item {\em Curse of dimensionality in the number of parameters $p$}.  It is well known that kernel density estimation breaks down as $p$ increases, with the sample size (number of posterior samples in this case) needing to increase exponentially in $p$ to maintain the same level of approximation accuracy.
\item {\em Subset posterior disconnection}. Because of the product form, the performance of the approximation to the full data posterior depends on the overlapping area of the subset posteriors, which is most likely to be the tail area of a subset posterior distribution. Hence, slight deviations in light-tailed approximations to the different subset posteriors might lead to poor approximation of the full data posterior. (See the right figure in Fig.\ref{fig:0})
\item {\em Mode misspecification}. For a multimodal posterior, averaging (noticing that the component mean of the kernel smoothing method is the average of subset draws) can collapse different modes, leading to unreliable estimates.
\end{enumerate}
\par
To ameliorate these problems, we propose a different method for parallelizing MCMC. This new method, designated as the {\em Weierstrass sampler}, is motivated by the Weierstrass transform, which is related to kernel smoothing but from a different perspective.  Our approach has good performance including in cases in which large sample normality of the posterior does not hold. In the rest of the article, we will use the term $f(\theta)$ to denote general posterior distributions and $f_i(\theta)$ for subset posteriors, in order to match the typical notation used with Weierstrass transform.
\par
 The key of our Weierstrass sampler lies in the Weierstrass transform,
\begin{align*}
  W_h f(\theta) = \int_{-\infty}^\infty \frac{1}{\sqrt{2\pi}h}e^{-\frac{(\theta-t)^2}{2h^2}}f(t) dt,
\end{align*}
which was initially introduced by \cite{Weier:1885}. The original proof shows that $W_h f(\theta)$ converges to $f(\theta)$ pointwise as $h$ tends to zero. Our method approximates the subset densities via the Weierstrass transformation.  We avoid inheriting the problems of the kernel smoothing method by directly modifying the targeted sampling distribution instead of estimating it from the subset posterior draws. 
In particular, we attempt to directly sample the approximated draws from the transformed densities instead of the original subset posterior distributions. Applying the Weierstrass transform to all subset posteriors and denoting $p(\theta|X_i)$ by $f_i(\theta)$, the full set posterior can be approximated as
\begin{align*}
 \notag \prod_{i=1}^m f_i(\theta)& \approx \prod_{i=1}^m W_{h_i} f_i(\theta) = \prod_{i=1}^m \int \frac{1}{\sqrt{2\pi}h_i}e^{-\frac{(\theta-t_i)^2}{2h_i^2}}f(t_i) dt_i\\
 \propto& \int \exp\bigg\{ -\frac{(\theta-\bar t)^2}{2h_0^2}\bigg\} \exp\bigg( -\frac{\bar {t^2} - \bar t^2}{2h_0^2}\bigg) f_1(t_1)\cdots f_m(t_m) dt,
\end{align*}
where $h_0^{-2} = \sum_{i=1}^m h_i^{-2}$, $w_i = h_i^{-2}/h_0^{-2}$, $\bar{t^2} = \sum_{i=1}^m w_it_i^2$ and $\bar t = \sum_{i=1} w_it_i$. The above equation can be viewed as the marginal distribution of the random variable $\theta$, derived from its joint distribution with the random variables $t_i, i=1,2,\ldots m$ with joint density
\begin{align}
   \exp\bigg\{ -\frac{(\theta-\bar t)^2}{2h_0^2} \bigg\} \exp\bigg(-\frac{\bar {t^2} - \bar t^2}{2h_0^2}\bigg)\cdot f_1(t_1)\cdot f_2(t_2)\cdots f_m(t_m). \label{eq:gauss}
\end{align}
The original subset posteriors appear as components of this joint distribution, enabling subset-based posterior sampling.  Moreover, the conditional distribution of the target random variable $\theta$ given the subset random variables is simply Gaussian. 

\section{Preliminaries}
The original Weierstrass transformation was stated in terms of the Gaussian kernel. For flexibility, we relax this restriction by broadening the choice of kernel functions, justifying the generalizations through Lemma 1 (See appendix). Following from the previous section,  the posterior density can be approximated as
\begin{align}
 \notag \prod_{i=1}^m f_i(\theta)& \approx \prod_{i=1}^m W_{h_i}^{(\textbf{K})}f_i(\theta) = \prod_{i=1}^m \int K_{h_i}(\theta - t_i)f_i(t_i) dt_i\\
 =& \int \prod_{i=1}^m h_i^{-1}K\big\{ h_i^{-1}(\theta - t_i)\big\} f_1(t_1)f_2(t_2)\cdots f_m(t_m) dt, \label{eq:approx}
\end{align}
where the last term can be viewed as the marginal distribution of random variable $\theta$, derived from its joint distribution with the random variables $t_i, i=1,2,\cdots m$ and a joint density $\prod_{i=1}^m K_{h_i}(\theta - t_i)f_i(t_i)$. If this joint density is proper (illustrated in Theorem \ref{thm:1}), one may sample $\theta$ from this distribution, utilizing the subset samplers as components. The details will be discussed in the next section. 

This section will focus on quantifying the approximation error of \eqref{eq:approx}. The results will be stated in terms of both one-dimensional and multivariate models.  The detailed derivation is only provided in the one-dimensional case in the Appendix, as the multivariate derivation proceeds along identical lines.
\par
A H\"older $\alpha$ smooth density function (for definition see Lemma 1) is always bounded on $\mathcal{R}^p$ for $\alpha\geq 0$. Let $M = \max_{i=1,2,\cdots,m}\{f_i(\theta), x\in \mathcal{R}^p\}$ denote the maximum value of the subset posterior densities. We have the following result (for the one-dimensional case).
\begin{theorem}\label{thm:1}
  If the posterior densities and the kernel functions satisfy the condition in Lemma 1 with $\alpha\geq 2$ and $k = 2$, i.e., the posterior density is at least second-order differentiable and the kernel function has finite variance, then the distribution defined in \eqref{eq:approx} is proper and there exists a positive constant $c_0$ such that when $h^2 = \sum_{i=1}^m h_i^2\leq c_0^{-1}$, the total variation distance between the posterior distribution and the approximation follows
  \begin{align*}
    \|f - \tilde f\| = \bigg\| C^{-1}\prod_{i=1}^m f_i(\theta) -C_W^{-1} \prod_{i=1}^m W_{h_i}^{(K)} f_i(\theta) \bigg\|\leq 2r_0r_1^{-2}h^2,
  \end{align*}
where $C$ and $C_W$ are the normalizing constants, and $r_0, r_1$ are defined as
\begin{align*}
  r_0  = C^{-1}M\max_{i\in\{1,2,\cdots,m\}}\int \prod_{j\neq i}^{m} f_j(\theta)dx \qquad  r_1^2 = \frac{2M}{M_2 \int t^2K(t)}.
\end{align*}
\end{theorem}
{\color{black} For multivariate distributions, the kernel variance $h_i$ should be substituted by the kernel covariance $H_i$ and we have a similar result.}
\begin{corollary} \label{corollary1}
  (Multivariate case) If the p-variate posterior densities $f_i(\theta_{i1},\cdots,\theta_{ip})$ and the kernel functions satisfy the conditions in Lemma 2 with $\alpha\geq 2$ and $k =2$, then for sufficiently small $h^2$, where $h^2 = \sum_{i=1}^mtr(H_i)$, the total variation distance between the posterior and the approximated density follows
  \begin{align*}
    \|f - \tilde f\| = \bigg\| C^{-1}\prod_{i=1}^m f_i(\theta_{i1},\cdots,\theta_{ip}) -C_W^{-1} \prod_{i=1}^m W_{H_i}^{(\textbf{K})} f_i(\theta_{i1},\cdots,\theta_{ip})\bigg\|\leq 2r_0r_1^{-2}h^2,
  \end{align*}
where $C,C_W$ are normalizing constants, and $r_0,r_1$ are defined in Theorem 1. 
\end{corollary}
In the error bound in Theorem 1, the constants $r_0$ and $c_0$ do not vary much even as the sample size increases to infinity. In fact, they will converge to constants in probability (see discussions in the Appendix after the proof of Theorem \ref{thm:1}). As a result, the choice of the tuning parameters $h_i$ is independent of the sample size, which is a desirable property.

 \section{Weierstrass refinement sampling}
As the error characterized by Theorem \ref{thm:1} maintains a reasonable level, the approximation \eqref{eq:approx} will be effective. In this section, a refinement sampler is proposed. For convenience we will focus on the Gaussian kernel if not specified otherwise, but modifications to other kernels are straightforward.  The algorithm is referred to as a {\em Weierstrass refinement sampler}, because samples from an initial rough approximation to the posterior (obtained via Laplace, variational approximations or other methods) are {\em refined} using information obtained from parallel draws from the different subset posteriors within a Weierstrass approximation.  Typically, the initial rough approximations can be obtained using parallel computing algorithms; for example, Laplace and variational algorithms are parallelizable.    
\par
Equation \eqref{eq:approx} can be used to obtain a Gibbs sampler. For the Gaussian kernel, the density in \eqref{eq:approx} can be reformulated into \eqref{eq:gauss}, which can then be used to construct a Gibbs sampler as follows (univariate case):
\begin{align}
  \theta|t_i &\sim N(\bar t, h_0^2) \label{eq:x}\\
  t_i|\theta &\sim \frac{1}{\sqrt{2\pi}h_i}e^{-\frac{(t_i-\theta)^2}{2h_i^2}} f_i(t_i)\quad i = 1,2,\ldots,m\label{eq:t}.
\end{align}
This approach takes the parameters in each subset as latent variables, and updates $\theta$ via the average of all latent variables. The Gibbs updating is used as a refinement tool; that is, the parameters $\theta$ are initially drawn from a rough approximation (Laplace, variational) and then plugged into the Gibbs sampler for one step updating, known as a refinement step.  Theorem 2 shows refinement leads to geometric improvement.
\begin{theorem}\label{thm2}
  Assume $\theta_0 \sim f_0(\theta_0)$ which is an initial approximation to the true posterior $f(\theta)$. By doing one step Gibbs updating as described in \eqref{eq:x}, \eqref{eq:t} (with general kernel $K$), we obtain a new draw $\theta_1$ with density $f_1(\theta)$. Using the notations in Theorem 1, if the kernel density function $K$ is fully supported on $\mathcal{R}$, then for any given $p_0\in (0, 1)$, there exists a measurable set $D$ such that $\int_D \tilde f(x)dx > p_0$ and,
\begin{align*}
  \int |f_1(\theta) - \tilde f(\theta)|d\theta \leq (1 - \eta)\int_D|f_0(\theta) - \tilde f(\theta)|d\theta + \int_{D^c} |f_0(\theta) - \tilde f(\theta)|d\theta,
\end{align*}
where $\eta$ is a positive value depending on $p_0$ but independent of $f_0$.
\par
Furthermore, if the kernel function satisfies the following condition,
  \begin{align}
    \lim_{\theta\rightarrow\pm\infty}\inf \frac{K_h(\theta - t)}{W_h^{(K)} f_i(\theta)} > 0 \label{eq:thm2}
  \end{align}
for given $t\in \mathcal{R}$ and $i\in\{1,2,\cdots,m\}$, then the total variational distance follows
  \begin{align*}
    \|f_1(\theta_1) - \tilde f(\theta)\| \leq (1 - \eta) \|f_0(\theta_0) - \tilde f(\theta)\|
  \end{align*}
for some $\eta > 0$, which is independent of $f_0$.
\end{theorem}
 The proof of Theorem \ref{thm2} is provided in the appendix. {\color{black} Though the theorem is stated for univariate distributions, the result is applicable to multivariate distributions as well.} The concrete refinement sampling algorithm is described below as Algorithm \ref{alg:WRS}.
\begin{algorithm}[htb]
\caption{Weierstrass refinement sampling}
\label{alg:WRS}
\begin{algorithmic}[1]
\REQUIRE ~~\\   
\STATE Input $H_i$ {\color{black} (or $h_i$ for univariate case)} for $i = 1,2,\cdots, m$.
\STATE $N = \mbox{ number of samples }, MCMC\leftarrow\{\}, ~~T_{i}\leftarrow \{\}, i=1,2,\cdots,m$; 
\\
\FOR{$k = 1$ \TO $N$}
\STATE $\theta_k\sim \hat f(\theta)$; ~~\COMMENT{$\hat f(\theta)$ is the rough approximation for the full set posterior}\\
\ENDFOR
\\
~\\
\ENSURE ~~\\                           
\STATE \COMMENT{On each of the m subsets $i=1,2,\cdots,m$ in parallel,}
\FOR{$k = 1$ \TO $N$ } \label{WRS:6}
\STATE Sample $t_{i}^{(k)} = (t_{i1}^{(k)},\cdots,t_{ip}^{(k)})^T \sim dN(t_i|\theta_k, H_i)\cdot f_i(t_i)$;
\STATE $T_{i}\leftarrow T_{i}\cup\{t_i^{(k)}\}$;
\ENDFOR \label{WRS:9}
\\
~\\
\STATE\COMMENT{Collect all $T_i, i=1,2,\cdots,m$ to draw posterior samples}
\FOR {$k=1$ \TO $N$}
\STATE $\tilde \theta \sim N(m^{-1}\sum_{i=1}^m t_{i}^{(k)},(\sum_{i=1}^m H_i^{-1})^{-1})$;
\STATE $MCMC\leftarrow MCMC\cup\{\tilde \theta\}$;
\ENDFOR

\RETURN{MCMC}
\end{algorithmic}
\end{algorithm}
\par
There are a number of advantages of this refinement sampler. First, the method addresses the dimensionality curse (issue 1 in Section 2, which appears as the inefficiency of the Gibbs sampler \eqref{eq:t} and \eqref{eq:x} when $h$ is small), with the large sample approximation only used as an initial rough starting point that is then refined.  We find in our experiments that we have good performance even when the true posterior is high-dimensional, multimodal and very non-Gaussian.
Second, as can be seen from \eqref{eq:t}, the subset posterior densities are multiplied by a common conditional component, which brings them close to each other and limits the problem with subset posterior disconnection (issue 2 of Section 2).  In addition, step \ref{WRS:6}-\ref{WRS:9} can be fully parallelized, as each draw is an independent operation.  There may be advantages of an iterative version of the algorithm, which runs more than one step of the Gibbs update on each of the initial draws (repeating Algorithm \ref{alg:WRS} several times). 
\par
The choice of tuning parameters $H_i$ and the relationship with number of subsets $m$ is an important issue. The parameter $H_i$ on the one hand controls the approximation error for each subset posterior. Apparently, $tr(H_i)$ has to be reduced as the number of subsets $m$ increases. On the other hand, $H_i$ also determines the efficiency of the Gibbs sampler (how fast can the Gibbs sampler evolve the initial approximation towards the true posterior), thus, $H_i$ might need to be chosen adequately large for efficient refinement. Such an argument leads to the conclusion of changing $H_i$ during the refinement sampling process if the refinement will be repeated multiple times. 
\par
According to Fukunaga's \citep{Fukunaga:1972} approach in choosing the bandwidth, if we attempt to apply kernel density estimation directly to the posterior distribution obtained from the full data set, the optimal choice will be
\begin{align}
  H_0 =\big\{(p+2)/4\big\}^{-2/(p+4)} N^{-2/(p+4)}\hat \Sigma, \label{eq:H0}
\end{align}
where $p$ is the number of parameters, $N$ is the total number of posterior samples and $\hat\Sigma$ is the sample covariance of the posterior distribution (to be approximated by the inverse of the Hessian matrix at mode). Based on this result, the starting value for $H_i$ could be chosen as $mH_0$, of which the magnitude is comparable to the covariance of each subset posterior distribution, admitting efficient refinement in the beginning. As the refinement proceeds, the ending point of $H_i$ should be around $m^{-1}H_0$ indicating an accurate approximation. The tuning parameters in the middle of the process should be chosen in between. For example, for a 10-step refinement procedure, one could use $mH_0$ for the first three steps, $H_0$ for the next five steps and $m^{-1}H_0$ for the last two steps.

\section{Weierstrass rejection sampling}

Algorithm \ref{alg:WRS} requires an initial approximation to the target distribution. To avoid this initialization, we propose an alternative self-adaptive algorithm, which is designated as {\em Weierstrass rejection sampler}. Because of the self-adapting feature, the algorithm suffers from certain drawbacks, which will be discussed along with possible solutions in the latter part of this section. We begin with a description of the algorithm.
\par
The formula shown in \eqref{eq:gauss} immediately evokes a rejection sampler for sampling from the joint distribution: assuming $t_i\sim f_i, i=1,2,\ldots,m$. Since $\exp(-\frac{\bar{t^2} - (\bar t)^2}{2h_0^2})\leq 1$, we can accept a draw of $\theta \sim N(\bar t, h_0^2)$ with probability $\exp(-\frac{\bar{t^2} - \bar t^2}{2h_0^2})$. Such an approach makes use of the average of the draws from the subsets to generate further posterior draws, which is similar to the kernel smoothing method proposed by \cite{Neis:etal:2013}. However, with slight modification, the Weierstrass rejection sampler can use the subset posterior draws as approximated samples directly without averaging, and thus avoid the mode misspecification issue. The result is stated below (for univariate case).
\begin{theorem}\label{thm3}
  If the subset posterior densities $f_k, k=1,2,\ldots, m$ satisfy all the conditions stated in Theorem 1 with $\alpha\geq 2$ and posterior draws $\theta_k$, and the second order kernel function $K(\cdot)$ (which is a density function) satisfies that $\max_{\theta\in\mathcal{R}}K(\theta) \leq c$ for some positive constant $c$, then for any given $i\in\{1,2,\cdots,m\}$, the rejection sampler accepts $\theta_i$ with probability $c^{-m+1}\prod_{k\neq i}^m K (\frac{\theta_k-\theta_i}{h_k})$. Referring to the density of the accepted draws by $g(\theta)$, the total variation distance of the approximation error follows
  \begin{align*}
   \|g(\theta) - f(\theta)\| = \frac{1}{2}\bigg\| g(\theta) - C^{-1}\prod_{k=1}^m f_k(\theta)\bigg\|_{L_1}\leq  2r_0r_1^{-2} h'^2,
  \end{align*}
where $h'^2 = \sum_{k\neq i}^m h_k^2$. The constants $r_0, r_1$ are defined in Theorem 1.
\end{theorem}
The following corollary is the multivariate version.
\begin{corollary}
    (Multivariate case) If the p-variate posterior densities $f_i(\theta_{i1},\cdots,\theta_{ip}), i=1,2,\cdots,m$ satisfy the conditions in Corollary \ref{corollary1} with $\alpha\geq 2$ and posterior draws $\theta_k$, and the second order kernels $K_j, j=1,2,\cdots,p$ are bounded by positive constant $c$,  then for any given $i\in\{1,2,\ldots,m\}$, the rejection sampler accepts $\theta_i$ if $c^{-p(m-1)}\prod_{k\neq i}^m\prod_{j=1}^p K_j (\frac{\theta_{kj}-\theta_{ij}}{r_1h_{kj}}) \geq u$ where $u\sim Unif(0,1)$. Referring to the density of the accepted draws by $g(\theta)$, the approximation error follows,
  \begin{align*}
   \|g(\theta) - f(\theta)\| = \frac{1}{2}\bigg\| g(\theta) - C^{-1}\prod_{k=1}^m f_k(\theta)\bigg\|_{L_1}\leq  2r_0r_1^{-2} h'^2,
  \end{align*}
where $h'^2 = \sum_{k\neq i}^mtr(H_i)$. The constants $r_0, r_1$ are defined in Theorem 1.
\end{corollary}
With the above theorem, it is easy to construct a rejection sampler as follows: for each iteration we randomly select one draw from the pool (or according to some reasonable weights), and perform the rejection sampling according to Theorem \ref{thm3}. We repeat the procedure for all iterations and then gather the accepted draws. The reason that the rejection operation is only conducted on one draw within the same iteration is to avoid incorporating extra undesirable correlation between the accepted draws.
\par
The effectiveness of a rejection sampler is determined by the acceptance rate. For a $p$-variate model and $m$ subsets, the acceptance rate for the Weierstrass rejection sampler can be calculated as (assuming all $h_{kj}$s are equal to $h^*$),
\begin{align}
  AR_{p,m} = P\bigg\{ c^{-p(m-1)}\prod_{k\neq i}^m\prod_{j=1}^p K_j\bigg(\frac{\theta_{kj}-\theta_{ij}}{r_1h_{kj}}\bigg) \geq u\bigg\} = O\bigg(\frac{r_1h^*}{c} \bigg)^{p(m-1)},
\end{align}
for adequately small $h^*$. Clearly, the acceptance rate suffers from the curse of dimensionality in both $m$ and $p$, so the number of posterior samples has to increase exponentially with $m$ and $p$.  This is the same problem as with the kernel smoothing method discussed before. To ameliorate the dimensionality curse, we provide the following solution.

\subsection{Sequential rejection sampling}
The number of subsets $m$ is easier to tackle. It is straightforward to bring $m$ down to $\log_2 m$ by using a {\bf pairwise combining strategy}: we first combine the $m$ subsets into a pairwise manner to obtain posterior draws on $m/2$ subsets.  This process is repeated to obtain $m/4$ subsets and so on until obtaining the complete data set. The whole procedure takes about $\lfloor \log_2 m \rfloor$ steps, and thus brings the power from $m$ down to $\lfloor \log_2 m \rfloor$.

\par
The curse of dimensionality in the number of parameter $p$ is less straightforward to address. If the $p$-variate posterior distribution $f$ satisfies that the parameters are all independent, then
\begin{align*}
    f(\theta_1,\theta_2,\cdots,\theta_p) = \prod_{j=1}^p f(\theta_j)
\end{align*}
and
\begin{align}
f(\theta_j) \propto \prod_{i=1}^m f_i(\theta_{ij}), \label{eq:mag}
\end{align}
where $f(\theta_j)$ and $f_i(\theta_{ij})$ are the posterior marginal densities for the $j^{th}$ parameter on the full set and on the $i^{th}$ subset, respectively.
The equation \eqref{eq:mag} indicates that the posterior marginal distribution satisfies the independent product equation as well. Therefore, one can obtain the joint posterior distribution from the marginal posterior distributions, which are obtained by combining the subset marginal posterior distributions via Weierstrass rejection sampling. This approach thus avoids the dimensionality issue (since in this case $p$ is always equal to 1). 
\par
 In general, the posterior marginal distribution does not satisfy the equation \eqref{eq:mag} because,
\begin{align*}
  f(\theta_j) &= \int f(\theta_1,\theta_2,\cdots,\theta_p)\prod_{k\neq j}d\theta_k \propto \int \prod_{i=1}^m f_i(\theta_{i1},\cdots,\theta_{ip})\prod_{k\neq j} d\theta_{ik}\\
&= \prod_{i=1}^m f_i(\theta_{ij}) \cdot  \int \prod_{i=1}^m p_i(\theta_{ik},k\in\{1,2,\cdots,p\}\backslash \{j\}|\theta_{ij}).
\end{align*}
The first term in the above formula is exactly the independent product equation, while the second term is due to parameter dependence. However, inspired by this marginal combining procedure, we propose a (parameter-wise) sequential rejection sampling scheme that decomposes the whole sampling procedure into $p$ steps, where each step is a one-dimensional conditional combining. The intuition is as follows: We first sample $\theta_1$ from its subset marginal distribution $f_i(\theta_{i1}), i = 1,2,\cdots, m$, and combine the draws via the Weierstrass rejection sampler to obtain $\theta_1^*\sim \prod_{i=1}^m f_i(\theta_{i1})/C_0$. Next, we plug in $\theta_1^*$ into each subset likelihood to sample $\theta_2$ from its subset conditional distribution $f_i(\theta_{i2}|\theta_1^*), i = 1,2,\cdots,m$ and then combine them to obtain the $\theta_2^* \sim \prod_{i=1}^m f_i(\theta_2|\theta_1^*)/C_1$, where $C_1$ is the normalizing constant (Notice that $C_1$ depends on the value of $\theta_1^*$). We continue the procedure until $\theta_p^*$, obtaining one posterior draw $\theta^* = (\theta_i^*, i=1,2,\cdots,p)$ that follows
\begin{align}
  \theta^* \sim \frac{\prod_{i=1}^m f_i(\theta_1)f_i(\theta_2|\theta_1)\cdots f_i(\theta_p|\theta_j, j<p)}{\prod_{j=0}^{p-1} C_j} = \frac{\prod_{i=1}^m f_i(\theta_j, j=1,\cdots,p)}{\prod_{j=0}^{p-1} C_j}. \label{eq:psis}
\end{align}

The numerator of \eqref{eq:psis} is exactly the target formula \eqref{eq:IPE}, while the denominator serves as the importance weights. The advantage of this new scheme is that it eliminates the dimensionality curse while keeping the number of required sequential steps low ($p$ steps). Moreover, the importance weights can be calculated easily and accurately as
\begin{align}
  C_j = \int \prod_{i=1}^m f_i(\theta_{j+1}|\theta_k^*, k\leq j) d\theta_{j+1}.
\end{align}
Notice that the integrated functions are all one-dimensional. Thus, an estimated (kernel-based) density $\hat f_i(\theta_{j+1}|\theta_k^*, k\leq j)$, combined with the numerical integration technique, is adequate to provide an accurate evaluation for $C_j$. An alternative approach for estimating $C_j$ is as follows,
\begin{align}
  C_j =  \frac{f(\theta_{j+1}^*|\theta_k^*, k\leq j)}{\prod_{i=1}^m f_i(\theta_{j+1}^*|\theta_k^*, k\leq j)},
\end{align}
where $f(\theta_{j+1}|\theta_k^*, k\leq j)$ can be obtained from kernel density estimation on the combined draws of $\theta_{j+1}^*$ (which requires to sample more than one $\theta_{j+1}^*$ at each iteration). $f_i(\theta_{j+1}|\theta_k^*, k\leq j)$ can be obtained similarly as before from the subset draws. The whole scheme is described in Algorithm \ref{alg:PSIS}.
\begin{algorithm}[htb]
\caption{(Sequential rejection sampling}
\label{alg:PSIS}
\begin{algorithmic}[1]
\REQUIRE ~~\\   
\STATE Input $N_0, h_j, j=1,2,\cdots, p$.\\
 ~~~~\COMMENT{$N_0$ is the number of samples drawn within each of the $p$ steps.}
\STATE Set $h_{ij} = \sqrt{m}h_j$ for $i = 1,2,\cdots, m$ and $C_j = 0, j=0,\cdots, p-1$.
\STATE $MCMC\leftarrow\{\}, ~~T_{i}\leftarrow \{\}, i=1,2,\cdots,m$; 
\\
\ENSURE ~~\\                           
\FOR{$j = 1$ \TO $p$ }
\STATE \COMMENT{On each of the m subsets $i=1,2,\cdots,m$,}
\FOR{$t = 1$ \TO $N_0$}
\STATE Sample $\theta_{ij}^{(t)}\sim f_i(\theta_j|\theta_1^*,\cdots,\theta_{j-1}^*)$
\STATE $T_{i}\leftarrow T_{i}\cup\{\theta_{ij}^{(t)}\}$;
\ENDFOR
\STATE\COMMENT{Collect all $T_i, i=1,2,\cdots,m$ to combine the draws}
\STATE Obtain one $\theta_j^*$ by combining $\theta_{ij}^{(t)}, i=1,2,\cdots,m$ via Weierstrass rejection sampling. \label{alg:rep}
\STATE Calculate $C_{j-1}$ as $C_{j-1} = \int \prod_{i=1}^m \hat f_i(\theta_j|\theta_1^*, \cdots,\theta_{j-1}^*) d\theta_j$.
\STATE $MCMC \leftarrow MCMC\cup\{(\theta_j^*, C_{j-1})\}$.
\ENDFOR
\\

\RETURN{MCMC}
\end{algorithmic}
\end{algorithm}
\par
Algorithm \ref{alg:PSIS} only produces one simulated posterior draw. Therefore, in order to obtain a certain number of posterior draws, the algorithm needs to be executed in parallel on multiple machines. For example, if one aims to acquire $N$ posterior draws, then $N$ parallel machines can be used, with each machine able to run $m$ sub-threads to fulfill the whole procedure. It is worth noting that this new scheme is also applicable to the kernel smoothing method proposed by \cite{Neis:etal:2013} for overcoming the dimensionality curse: just substituting the step \eqref{alg:rep} in Algorithm \ref{alg:PSIS} with the corresponding kernel method.
\par
The new algorithm still involves a sequential updating structure, but the number of steps is bounded by the number of parameters $p$, which is different from the usual MCMC updating. A brief interpretation for the effectiveness of the new algorithm is that it changes how error is accumulated. The original $p$-dimension function with a bandwidth $h$ will accommodate the error in a manner as $h^{2/p}$, while now with the decomposed p steps, the error is accumulated linearly as $ph^2$ which reduces the dimensionality curse.

\section{Numerical study}
In this section, we will illustrate the performance of Weierstrass samplers in various setups and compare them to other partition based methods, such as subset averaging and kernel smoothing. More specifically, we will compare the performance of the following methods.
\begin{itemize}
\item Single chain MCMC: running a single Markov chain Monte Carlo on the total data set.
\item Simple averaging: running independent MCMC on each subset, and directly averaging all subset posterior draws within the same iteration.
\item Inverse variance weighted averaging: running MCMC on all subsets, and carrying out a weighted average for all subset posterior draws within the same iteration. The weight follows
  \begin{align*}
    w_i = (\sum_{k=1}^m \hat\Sigma_i^{-1})^{-1}\hat \Sigma_i^{-1},
  \end{align*}
where $\hat \Sigma_i$ is the posterior variance for the subset $i$.
\item Weierstrass sampler: The detailed algorithms are provided in previous section. For the Weierstrass rejection sampler, we do not specify the value of the tuning parameter $h_{ij}$, but instead, we specify the acceptance rate. (i.e., we first determine the acceptance rate and then calculate the corresponding $h_{ij}$). 
\item Non-parametric density estimation (kernel smoothing): Using kernel smoothing to approximate the subset posteriors and obtain the product thereafter. Because this method is very sensitive to the choice of the bandwidth and the covariance of the kernel function when the dimension of the model is relatively high, in most cases, only the procedure of marginal subset densities combining will be considered in this section.

\item Laplacian approximation: A Gaussian distribution with the posterior mode as the mean and the inverse of the Hessian matrix evaluated at the mode as the variance.
\end{itemize}
The models adopted in this section include logistic regression, binomial model and Gaussian mixture model. The performance of the seven methods will be evaluated in terms of approximation accuracy, computational efficiency, and some other special measures that will be specified later. We made use of the R package ``BayesLogit'' \citep{Polson:etal:2013} to fit the logistic model and wrote our own JAGS code for the Gaussian mixture model.

\subsection{Weierstrass refinement sampling}
We assess the performance of the refinement sampler in this section. The first part will be an evaluation of the refinement property as claimed in Theorem \ref{thm2} and in the second part we will compare the performance for various methods within the logistic regression framework.
\subsubsection{Refinement property}
We evaluate the refinement property under both a bi-modal posterior distribution and the real data. For the bi-modal posterior distribution, the two subset posterior densities are
\begin{align*}
  p_1 = \frac{1}{2}N(-1.7, 0.5^2) + \frac{1}{2}N(0.8, 0.5^2)\quad p_2 = \frac{1}{2}N(-1.3, 0.5^2) + \frac{1}{2}N(1.2, 0.5^2).
\end{align*}
Then according to \eqref{eq:IPE}, the posterior on full data set will be roughly (omitting a tiny portion of probability),
\begin{align*}
  p_{12} = \frac{1}{2}N(-1.5, 0.5^2/2) + \frac{1}{2}N(1.0, 0.5^2/2).
\end{align*}
\par
The initial approximation adopted for the refinement sampling is a normal approximation to $p_{12}$ which has the same mean and variance. We trace the change of the refined densities for different numbers of iterations and illustrated them in Fig \ref{fig:ref} (left). Because the initial approximation is very different from the true target, the tuning parameter $h$ will start at a large value and then decline at a certain rate, in particular, $h = 0.8^{iteration}$. As shown in Fig \ref{fig:ref} (left), the approximation has evolved to a bi-modal shape after the first iteration and it appears that 10 steps are adequate to obtain a reasonably good result. 
\par
Fig \ref{fig:ref} (right) shows part of the refinement results for the real data set (which will be described in more detail in the real data section), in which a logit model was fitted on the 200,000 data set with a partition into 20 subsets. We set the initial approximation apart from the posterior mode and monitor the evolution of the refinement sampler (the tuning parameters are chosen according to \eqref{eq:H0}, and $H_0$ is chosen to be the inverse of the Hessian matrix at the mode). The refinement sampler quickly moves from the initial approximation to the truth in just 10 steps. (No posterior distribution from a full MCMC was plotted, as the sample size is too large to run a full MCMC)
\begin{figure}[!htpb]
  \centering
  \includegraphics[height = 6.5cm]{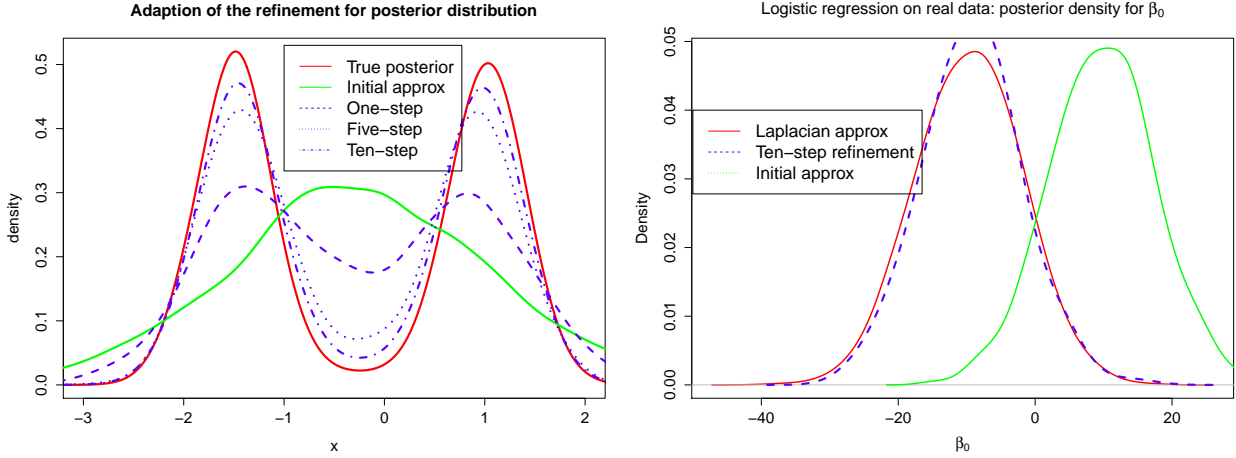}
  \caption{Refinement property for the Weierstrass sampler}
  \label{fig:ref}
\end{figure}

\subsubsection{Approximation accuracy}
We adopt the logistic model for assessing the approximation performance. The logistic regression model is broadly adopted in many scientific fields for modeling categorical data and conducting classification,
\begin{align*}
  P(Y=1|X) = logit^{-1}(X\beta+\beta_0) = \frac{\exp(X\beta+\beta_0)}{1+\exp(X\beta+\beta_0)},
\end{align*}
where $logit$ is the corresponding link function and $\beta_0$ is the intercept. 
The predictors $X = (X_1,X_2,\cdots,X_p)$ follow a multivariate normal distribution $N(0, \Sigma)$, and the covariance matrix $\Sigma$ follows
\begin{align*}
  \Sigma_{ii} = var(X_i) = 1, \quad \Sigma_{ij} = cov(X_i, X_j) = \rho,\quad i=1,2,\cdots,p
\end{align*}
where $\rho$ will be assigned two different values (0 and 0.3) to manipulate two different correlation levels (independent to correlated).
\par
In this study, the model contains $p = 50$ predictors and $n = 10,000$ or $ 30,000$ observations, both of which will be partitioned into $m = 20$ subsets. The coefficients $\beta$ follow
\[
\beta_i = \left\{ 
  \begin{array}{l l}
    (-1)^{u_i}(1+|N(0,1)|) & i\geq 10 \\
    0 & 1\leq i <10\\
    1 & i = 0
  \end{array}
\right.
\]
where $u_i$ is a Bernoulli random variable with $P(u_i=1) = 0.6$ and $P(u_i=-1) = 0.4$. The reason to specify the coefficients in this way is to demonstrate the performance of all methods in different situations (both easy and challenging). For logistic regression, the closer the coefficient is to zero, the larger the effective sample size and the better the performance of a Gaussian approximation. See Fig \ref{fig:0}.
\begin{figure}[!htpb]
  \centering
  \includegraphics[height = 17cm, angle = 270]{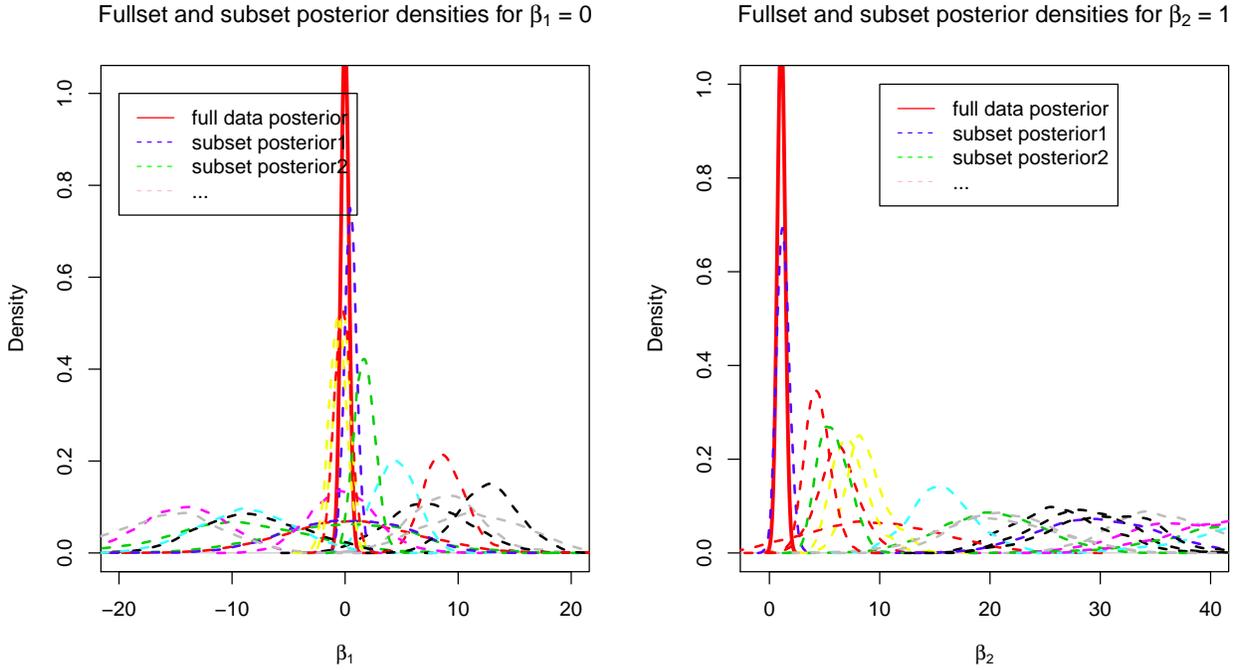}
  \caption{Subset posteriors for zero and non-zero coefficient. $\rho = 0, n = 10,000$.}
  \label{fig:0}
\end{figure}
\par
Because kernel density estimation is very sensitive to the choice of the kernel covariance when the dimension of the model is moderately high, {\color{black} we only demonstrate the performance of kernel smoothing for marginally combined posterior distributions \eqref{eq:mag} in this section. The Weierstrass rejection sampler is carried out in a similar way.} These two methods will be compared more formally in the next section. 
\par
50 synthetic data sets were simulated for each pair of $\rho$ and $n$. For posterior inference, we drew 20,000 samples (thinning to 2,000) after 50,000 burn-in for single MCMC chain and each subset MCMC. For Weierstrass refinement sampling, the Laplacian approximation is adopted as initial approximation, and the kernel variance is chosen according to \eqref{eq:H0}. We conduct 10 steps refinement to obtain 2,000 refined draws (within each refinement step, we run 100 MCMC iterations on each subset to obtain an updated draw). The results are summarized in Fig \ref{fig:logit_ind_10000}, \ref{fig:logit_03_10000}, \ref{fig:logit_ind_30000}, \ref{fig:logit_03_30000}, and Table \ref{tab:1}. 

\par
The posterior distribution of two selected parameters (including both zero and nonzero) for different $\rho$'s and $n$'s are illustrated in the figures. The nonzero parameter was plotted in two different scales in order to incorporate multiple densities in one plot.
The numerical comparisons include the difference of the marginal distribution of each parameter, the difference of the joint posterior and the estimation error of the parameters. We evaluate the difference of the marginal distribution by the average total variation difference (upper bounded by 1) between the approximated marginal densities and the true posterior densities. The result will be separately demonstrated for nonzero and zero coefficients, and denoted by $\|\hat p(\beta_1) - p(\beta_1)\|$ (nonzero) and $\|\hat p(\beta_0) - p(\beta_0)\|$ (zero) respectively. Evaluating the difference between joint distributions is difficult for multivariate distributions, as one needs to accurately estimate a distance between a true joint distribution and a set of samples from an approximation. Therefore, we adopted the approximated Kullback-Leibler divergence between two densities (approximating two densities by Gaussian) for reference,
  \begin{align*}
    D_{KL}(\hat p(\beta)||p(\beta)) = \frac{1}{2}(tr(\Sigma^{-1}\hat\Sigma)+(u-\hat u)^T\Sigma^{-1}(u-\hat u)-p-\log(|\hat\Sigma|/|\Sigma|)),
  \end{align*}
where $u$ and $\Sigma$ are the sample mean and sample variance of the true posterior $p(\beta)$, while $\hat u$ and $\hat \Sigma$ are those for $\hat p(\beta)$. Finally, the error of the parameter estimation will be demonstrated as the ratio between the estimation error of the approximating method and that of the true posterior mean, which is $Error(\beta_{a}|\beta_{p}) = \|\hat \beta_{approx}-\beta\|_2/\|\hat \beta_{posterior} - \beta\|_2$. The average results are shown in Table.\ref{tab:1}.
\begin{figure}[!htpb]
  \centering
  \includegraphics[height = 4.5cm]{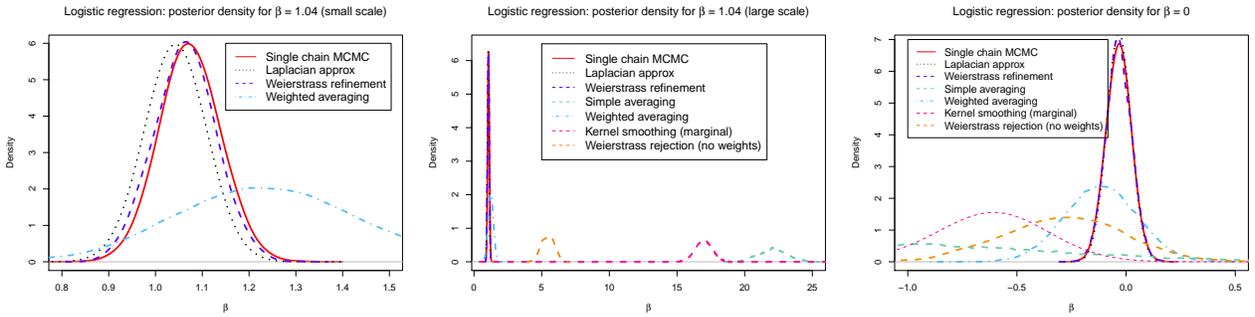}
  \caption{Posterior densities for independent predictors $\rho = 0, n = 10,000$.}
  \label{fig:logit_ind_10000}
\end{figure}

\begin{figure}[!htpb]
  \centering
  \includegraphics[height = 4.5cm]{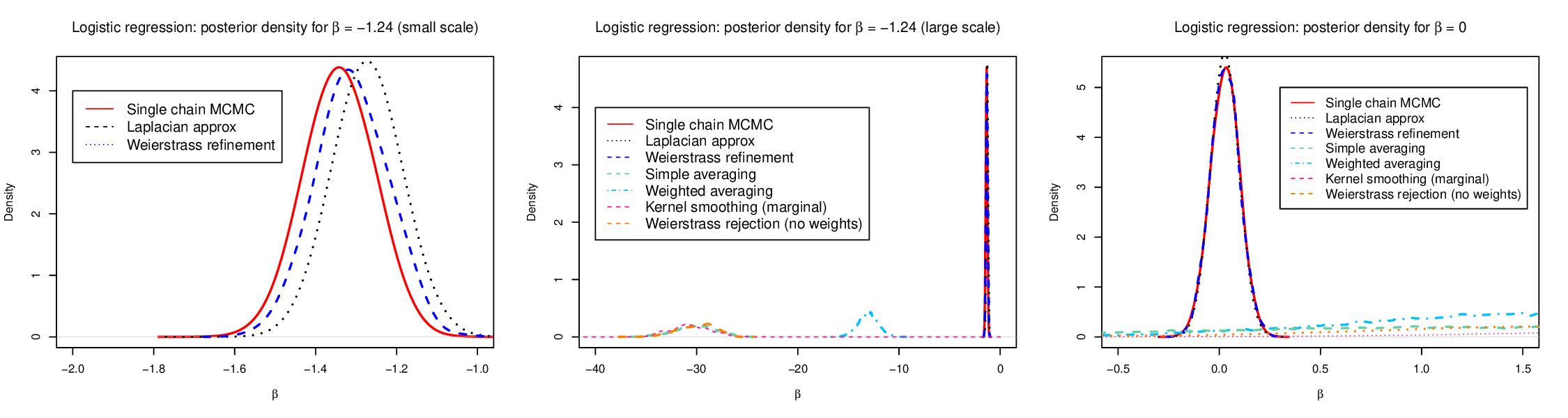}
  \caption{Posterior densities for independent predictors $\rho = 0.3, n = 10,000$.}
  \label{fig:logit_03_10000}
\end{figure}

\begin{figure}[!htpb]
  \centering
  \includegraphics[height = 4.5cm]{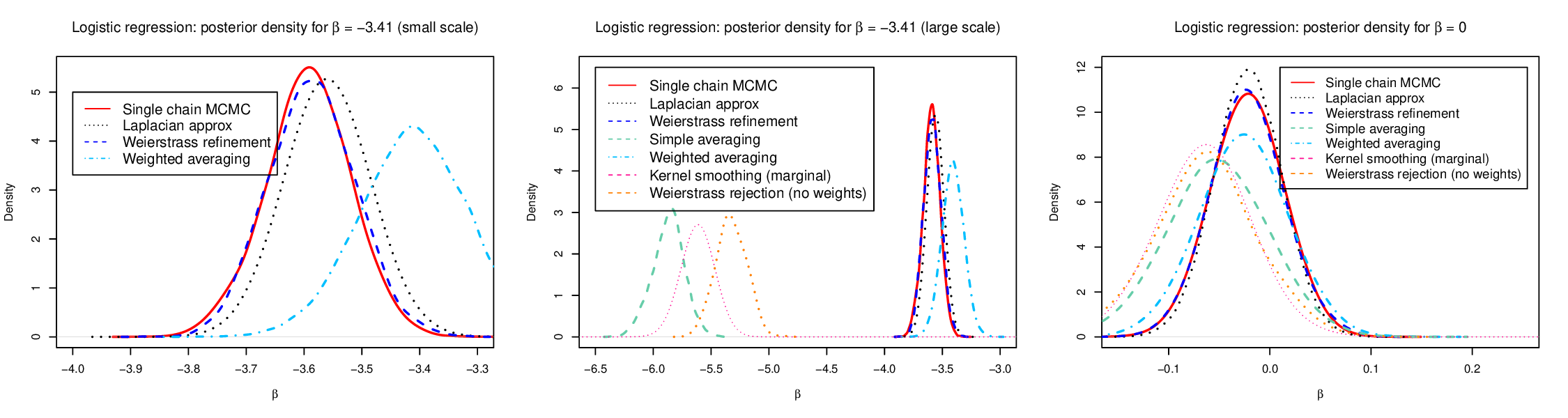}
  \caption{Posterior densities for independent predictors $\rho = 0, n = 30,000$.}
  \label{fig:logit_ind_30000}
\end{figure}

\begin{figure}[!htpb]
  \centering
  \includegraphics[height = 4.5cm]{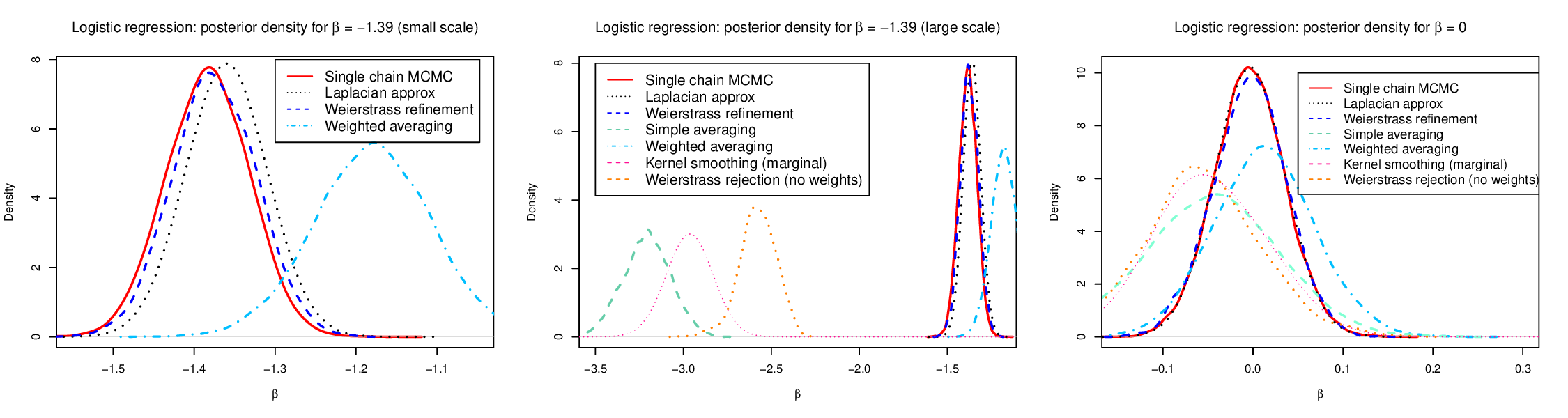}
  \caption{Posterior densities for independent predictors $\rho = 0.3, n = 30,000$.}
  \label{fig:logit_03_30000}
\end{figure}

\begin{table}[!htpb]
\footnotesize
  \centering
  \caption{Approximation accuracy for marginal and joint densities. W. Refi. and W. Rej stand for Weierstrass refinement sampling and rejection sampling (without weight correction).}
  \begin{tabular}{l|l|l|cccccc}
\hline\hline
                                     & $n$ &$\rho$ & W. Refi. & Laplacian & Simple Ave. & Weighted Ave. & Kernel & W.Rej\\
\hline
                                     & 10,000 & 0        & 0.0683    & 0.177 & 0.995 & 0.816 & 0.997 & 0.989\\
$\|\hat p(\beta_1) - p(\beta_1)\|$   &        & 0.3      & 0.105    & 0.238  & 1.00  & 0.873 & 1.00 & 1.00 \\
                                     & 30,000 & 0        & 0.0377   & 0.112 &  1.00   & 0.648  & 1.00 & 0.990\\
                                     &        & 0.3      & 0.0543    & 0.137 & 1.00 & 0.738 & 1.00 & 0.995\\

\hline
                                     & 10,000  & 0       &0.0306     &0.0157 & 0.923 & 0.619 & 0.904 & 0.826\\
$\|\hat p(\beta_0) - p(\beta_0)\|$   &         & 0.3     &0.0358     &0.0235 & 0.946 & 0.832 & 0.954 & 0.891\\
                                     & 30,000  & 0       &0.0231     &0.0091 & 0.268 & 0.106 & 0.255 & 0.245\\
                                     &         & 0.3     &0.0268     &0.0114 & 0.565 & 0.224 & 0.548 & 0.457\\

\hline
                                     & 10,000  & 0       &0.487       & 0.766 & $2.38\times 10^5$ & 313.16 & --- & ---\\
$D_{KL}(\hat p(\beta)||p(\beta))$     &        & 0.3      & 0.551       & 1.207  & $2.33\times 10^5$ & $2.32\times 10^4$ & --- & ---\\
                                     & 30,000  & 0        &0.359       & 0.463 & 478.47 &8.997 & --- & ---\\
                                     &         & 0.3      &0.419       & 0.604 & $2.07\times 10^4$ & 51.07 & --- & --- \\
\hline
                                     & 10,000  & 0       & 0.867      &0.619  & 352.11 & 9.20 & 293.30 & 197.86\\
 $Error(\beta_{a}|\beta_{p})  $&         &0.3      & 0.823      &0.391 & 362.90&144.02 & 249.74 &237.95 \\
                                     & 30,000  &0        & 0.922      &0.890 & 17.92 & 1.34 & 12.31 & 11.38\\
                                     &         &0.3      & 0.839      &0.602  & 151.13 & 3.99 & 102.20 & 26.04 \\
\hline
  \end{tabular}
\label{tab:1}
\end{table}

\subsection{Weierstrass rejection sampling}
In this section, we will evaluate the performance for the Weierstrass rejection sampler. Theorem 3 indicates that the rejection sampler may enjoy an advantage of posterior mode searching, as this sampler will make use of draws from subsets directly, instead of any form of averaging (averaging might mess up modes). This property will be investigated in the first part via the mixture model. For the second part, we compare the approximation accuracy of the rejection sampler and other methods through a simple binomial model.
\subsubsection{Mode exploration}
Finite mixture model are widely used for clustering and approximation, but face well known challenges due to non-identifiability and multi-modality. The Gibbs sampler for normal mixture models, as pointed out in \cite{Jasra:etal:2005}, suffers from the inefficiency of mode exploration, which has motivated methods such as split-merge and parallel tempering \citep{Earl:Deem:2005, Jasra:etal:2005, Neis:etal:2013}. In this section, we implement the problematic Gibbs sampler for normal mixture model without label switching moves, as a mimic to more general situations in which we do not have a specific solution for handling multi-modes (label switching is a specific method for dealing with mixture models). Then we parallelize this Gibbs sampler via different subset-based methods, and examine the abilities in posterior mode exploration. 
\par
The mixture distribution follows,
\begin{align*}
  x\sim \frac{1}{2}N(0, 0.5^2) + \frac{1}{4}N(2, 0.5^2) + \frac{1}{4} N(4, 0.5^2).
\end{align*}
We simulate 10,000 data points from this model, divided into 10 subsets, which will be analyzed via single chain MCMC as well as parallelized algorithms. Since the model is multidimensional, we drew 2000 samples via the sequential rejection sampling described in Algorithm \ref{alg:PSIS}. For other samplers, we obtained 20,000 posterior draws after 20,000 iterations burn-in on each subset. The results are plotted in Fig \ref{fig:mix3}.
\begin{figure}[!htpb]
  \centering
  \includegraphics[height = 5.2cm]{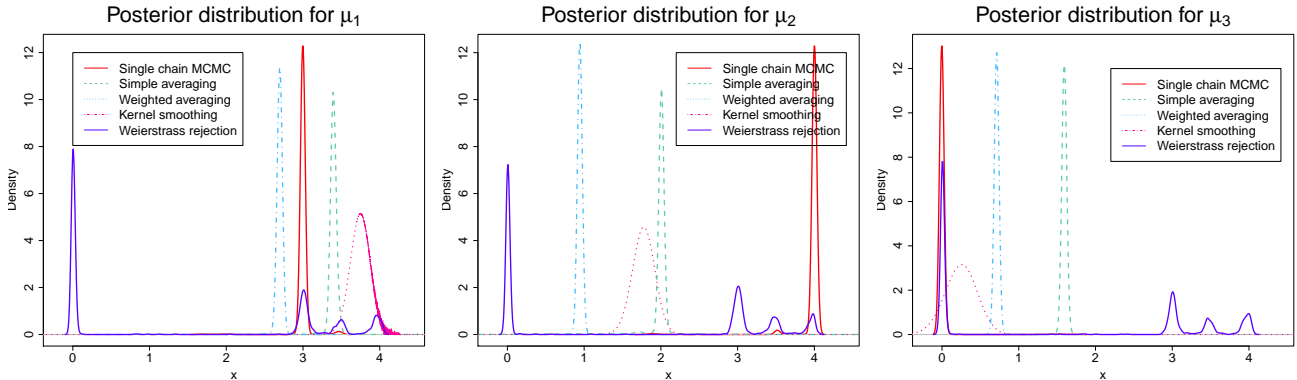}
  \caption{Posterior distributions for the component means}
  \label{fig:mix3}
\end{figure}
It can be seen from Fig \ref{fig:mix3} that the Weierstrass (sequential) rejection sampling correctly recognized the posterior modes. There is one false mode, but is also picked by the single chain MCMC. 

\subsubsection{Approximation accuracy}
In this section, the Beta-Bernoulli model is employed for testing the performance of the Weierstrass rejection sampler. The dimension of this model is low and the true posterior distribution is available analytically, which makes the model appropriate for comparing rejection sampling and kernel smoothing method. The parameter $p$ will be assigned two values in this section: $p = 0.1$ that corresponds to a common scenario and $p = 0.001$ which corresponds to the rare event case. We simulated 10,000 samples which were then partitioned into 20 subsets. To obtain a conjugate posterior, we assign a beta prior $Beta(0.01,0.01)$ and draw $100,000$ posterior samples for further analysis. The posterior densities for different values of $p$ are shown in Fig \ref{fig:5}.
\begin{figure}[!htbp]
  \centering
\includegraphics[height = 6.1cm]{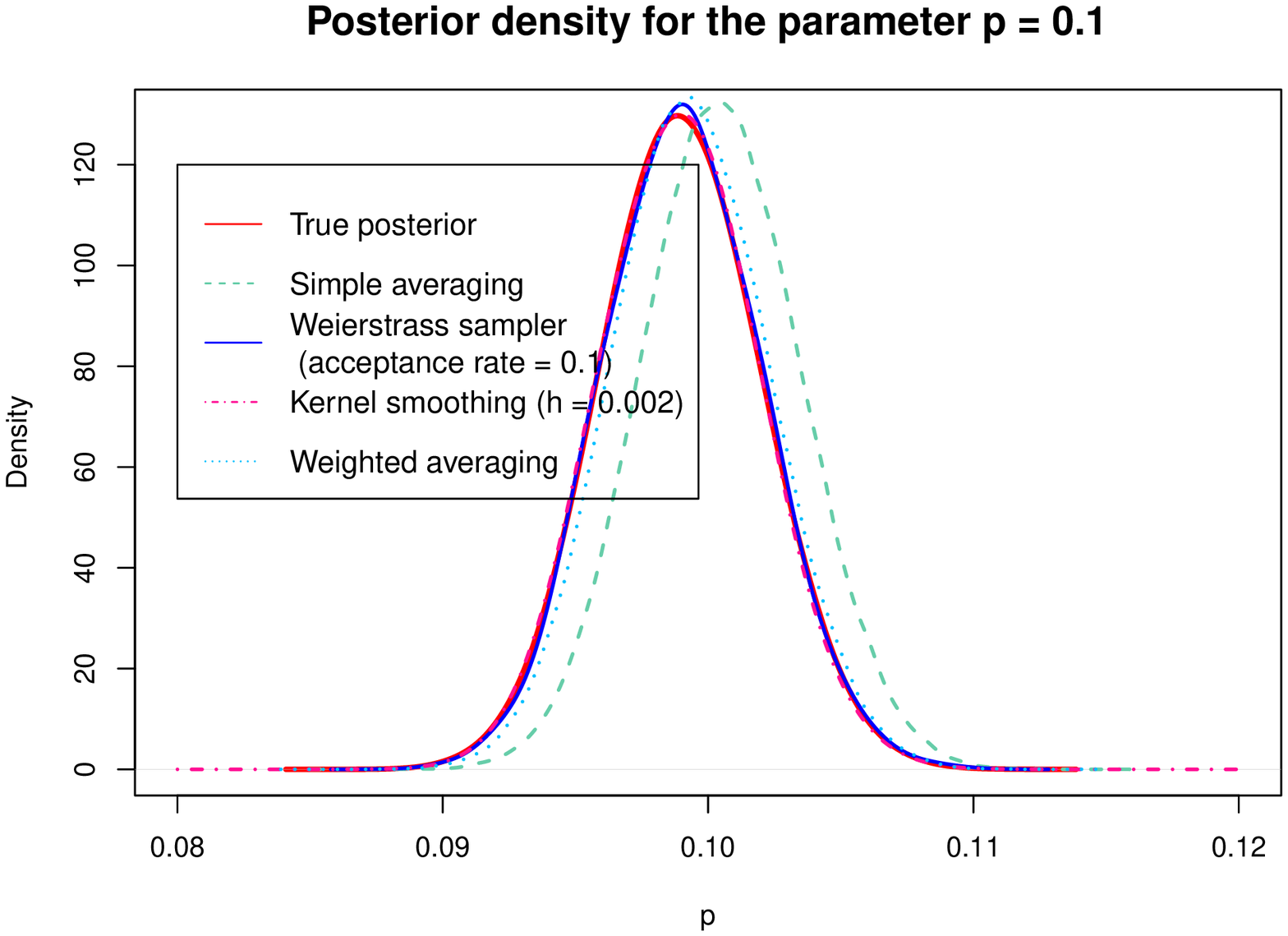}
\includegraphics[height = 6.1cm]{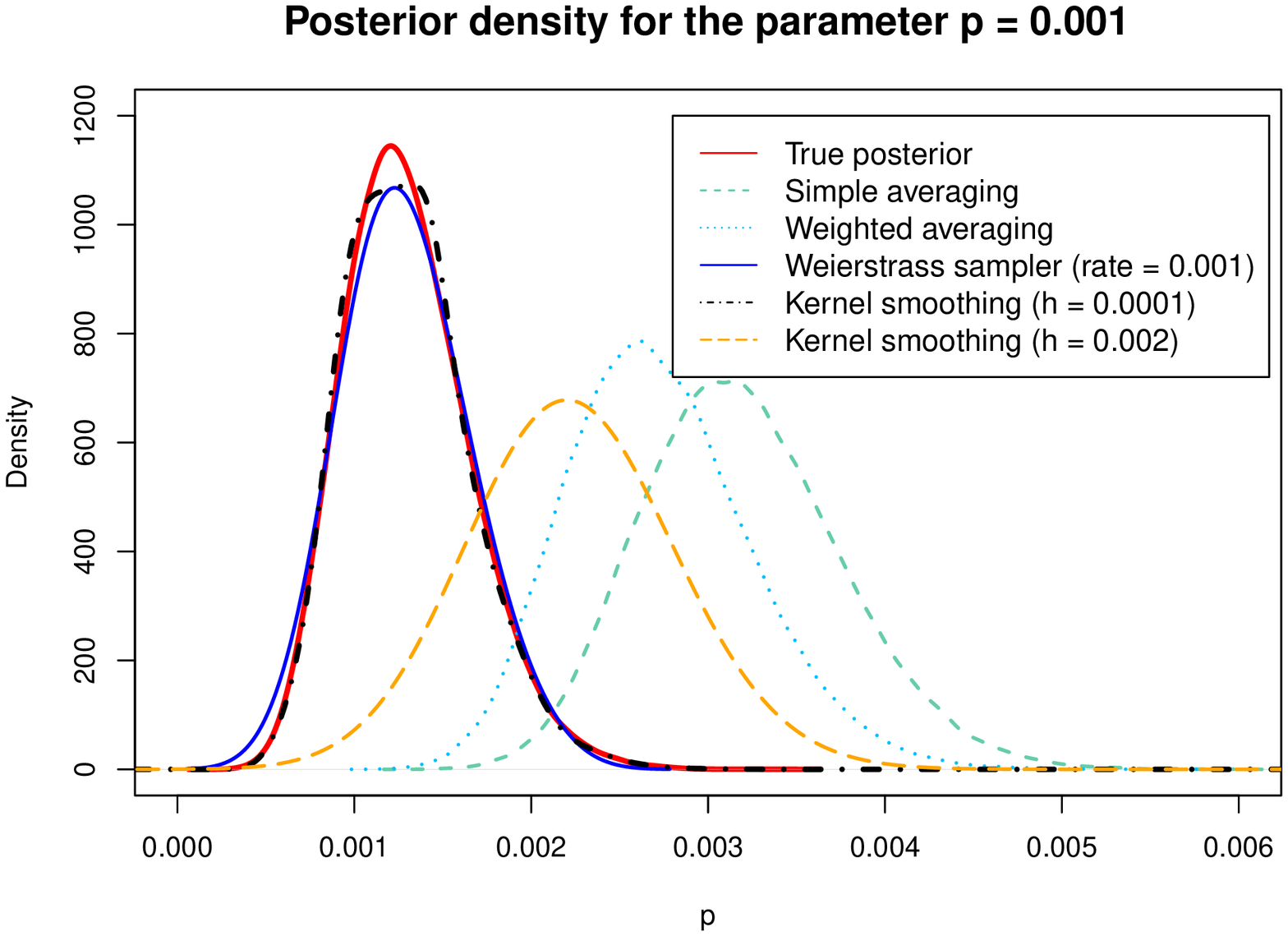}
  \caption{Posterior density for binomial data with $p = 0.1$ and $p = 0.001$. Top: the posterior densities for $p=0.1$. Bottom: the posterior densities for $p= 0.001$.}
 \label{fig:5}
\end{figure}
\par
It can be seen from Fig \ref{fig:5} that when the data set contains moderate amount of information (for $p = 0.1$), all methods work fine. However, for the inadequately informed case ($p = 0.001$), the kernel smoothing and Weierstrass rejection sampling are the only methods that perform appropriately. To match the scale of the posterior density, the rejection sampler adopts an acceptance rate of $0.1\%$. For kernel smoothing, it requires the bandwidth to be chosen as $h = 0.0001$ to achieve the same level of accuracy (See the result for larger $h$ in Fig.\ref{fig:5}).

\subsection{Computation efficiency}
Computational efficiency is a primary motivation for parallelization. Reductions in sample size lead to reduction in computational time in most cases. We demonstrate the effects of the total sample size and the number of parameters (or number of mixture components) for logistic regression and the mixture model in Fig \ref{fig:logit_time} and Fig \ref{fig:mix_time}. The subset number is fixed to 20. For each subset and the total set we drew 10,000 samples after 10,000 iterations burn-in. For Weierstrass refinement sampler, we repeat the procedure 10 times with 100 iterations within each step.
\begin{figure}[!htpb]
  \centering
  \includegraphics[height = 8.2cm]{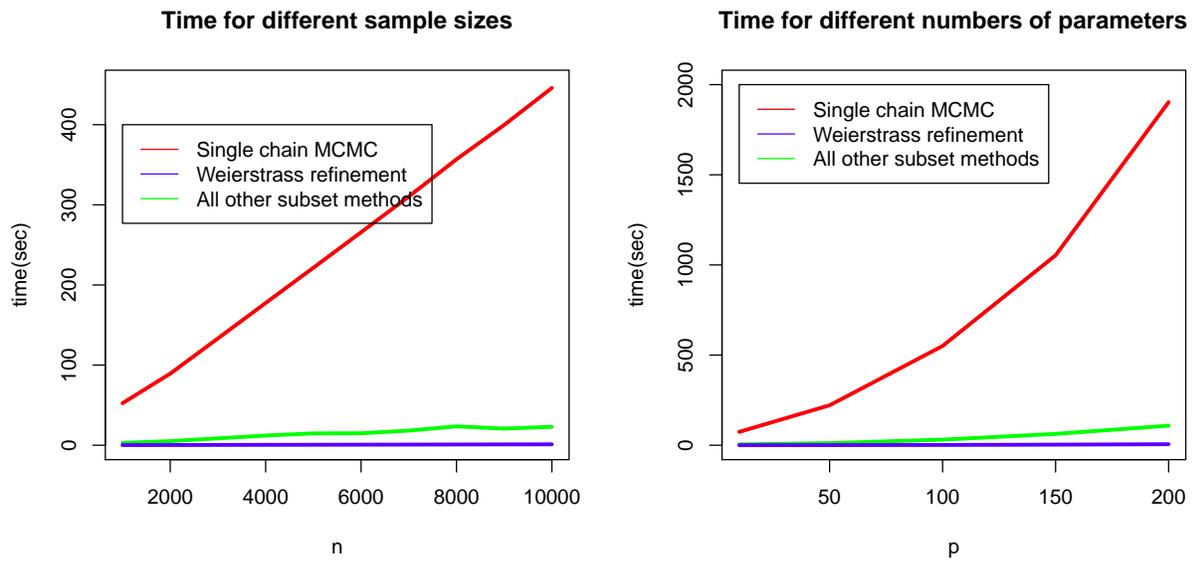}
  \caption{Computational time for logistic regression.}
  \label{fig:logit_time}
\end{figure}

\begin{figure}[!htpb]
  \centering
  \includegraphics[height = 8.2cm]{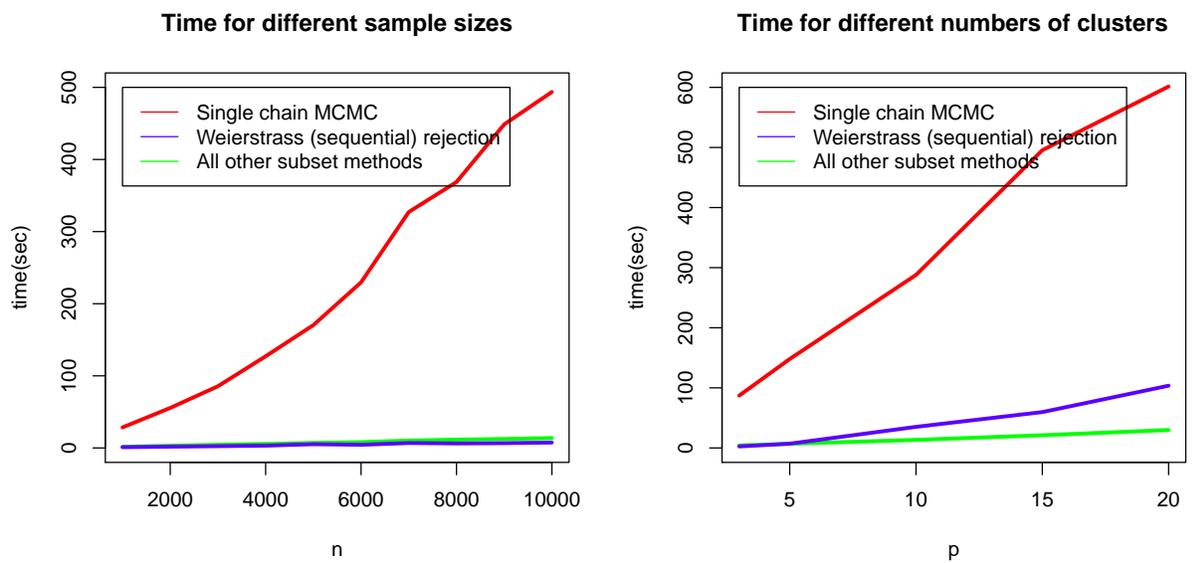}
  \caption{Computational time for the mixture model.}
  \label{fig:mix_time}
\end{figure}

\section{Real data analysis}
The real data set used in this section contains weighted census data extracted from the 1994 and 1995 current population surveys conducted by the U.S. Census Bureau \citep{Bache:Lichman:2013}. The purpose is to predict whether the individual's annual gross income will be higher than 50,000 USD (i.e., whether 50,000+ or 50,000-). The whole set contains around 200,000 observations, and 40 attributes which were turned into 176 predictors due to the re-coding of the categorical variables. Because the sample size is too large for fitting a logistic model via usual MCMC softwares such as JAGS or Rstan, we only illustrate the posterior inference for parallelized algorithms and the Laplacian approximation. In addition, because of the high dimension of the model and the sensitivity of kernel smoothing method, we only consider the marginal kernel combining algorithm and marginal Weierstrass rejection sampler. The latter will not be listed in the results as the performance is very similar to marginal kernel method.
\par
For posterior inference, we partition the data into 20 subsets, and drew 20,000 samples on each subset after 50,000 burn-in. For Weierstrass refinement sampler, the initial distribution is chosen to be slightly away from Laplacian approximation (See Fig \ref{fig:ref}, right figure) to avoid the potential unfair advantage. We conduct 10 step refinement on 2,000 initial draws, with 50 iterations within each step for proposing refined draws. The tuning parameter is chosen according to \eqref{eq:H0}. The Test Set 1 consists of 5,000 positive (50,000+) and 5,000 negative (50,000-) cases and the Test Set 2 contains 5,00 positive and 5,000 negative cases. Since the positive cases are rare (8\%) in the training set, we should expect different behaviors on the two different sets. For prediction assessment, the logistic model will predict positive if $logit^{-1}(E[x_i\beta])$ is greater than 0.5, and vice-versa. The posterior distribution of selected parameters for different methods were plotted in Fig \ref{fig:realdata}, and the prediction results for different categories and the two test data sets are listed in Table \ref{tab:realdata}.

\begin{figure}[!htpb]
  \centering
  \includegraphics[height = 7.2cm]{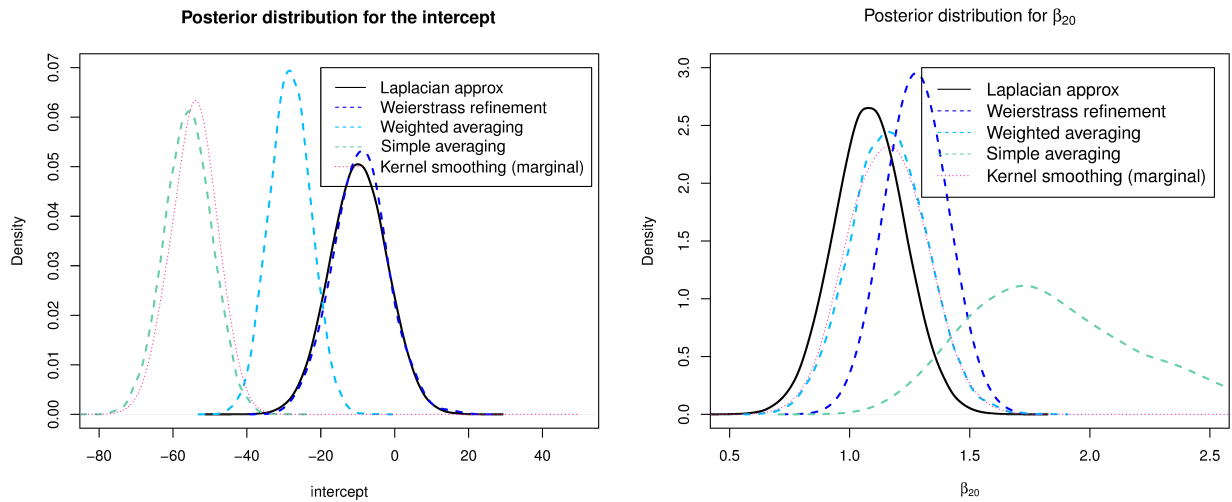}
  \caption{Posterior distribution for selected parameters.}
  \label{fig:realdata}
\end{figure}

\begin{table}[!htpb]
  \centering
\caption{Classification accuracy for test set}
  \begin{tabular}{l|cccc}
    \hline\hline
    Correctness on & 50,000- (\%) &  50,000+ (\%) & Test Set 1 (\%) & Test Set 2 (\%) \\
    \hline
    Weierstrass refinement & 97.0 & 57.9 & 77.4 & 93.0\\
    Laplacian approx  & 98.9 & 39.3 & 69.1 & 92.9\\
    Simple average & 98.9 & 39.4 & 69.0 & 92.9\\
    Weighted average & 99.0 & 39.4 & 69.2 & 93.0\\
    Kernel (marginal) &99.6 & 15.9 & 57.8 & 91.2\\
    \hline
  \end{tabular}
  \label{tab:realdata}
\end{table}
Because the positive category (annual income greater than 50,000 USD) is rare in the training set, Laplacian approximation is likely to overfit and the posterior might not be approximately Gaussian, leading to low accuracy in predicting positiveness for most methods. On the contrary, it can be seen that all methods perform well on the Test Set 2 (which mimics the ratio of the training set).

\section{Concluding remarks}
In this article, we proposed a new, flexible and efficient Weierstrass sampler for parallelizing MCMC. The Weierstrass sampler contains two different algorithms, which are carefully designed for different situations. Extensive numerical evidence shows that, compared to other methods in the same direction, Weierstrass sampler enjoys better performance in terms of approximation accuracy, chain mixing rate and a potentially faster speed. Faced with the same difficult issues, such as the dimensionality curse, Weierstrass sampler attempts to seize the balance in trading off between the accuracy and computation efficiency. As illustrated in the numerical study, the rejection sampler can not only well approximate the original MCMC, but also improve its performance in the posterior modes exploration. In the simulation, the sampler correctly identifies all {\color{black}the mixture components}, removing problems of the original Gibbs sampler.
\par
Future works of Weierstrass samplers may lie in the following aspects. First, investigating the asymptotic justification for the marginal combining strategy, which could help eliminate the dimensionality concern for both kernel density estimation method and Weierstrass rejection sampling. Second, investigating the potential application in parallel tempering. In parallel tempering, there is a temperature parameter $T$ which controls both the approximation accuracy and the chain exploration ability. Here, with the tuning parameter $h$, one is able to achieve the same thing: small $h$ entails a high accuracy, while large $h$ ensures a better exploration ability. {\color{black} Therefore, one could design a set of different values of $h$ for the sampling procedure, providing a `parallel' way of doing parallel tempering, which may potentially improve the performance of the original method.}

\section*{Appendix}
\subsection*{Proofs for preliminaries and Theorem 3}
\begin{lemma}\label{lemma:1}
  Assume a real-valued function $f$ is H\" older $\alpha$ differentiable with constant $C_0$, i.e., for $l = \lfloor \alpha \rfloor$, the $l$-th derivative of $f$ follows,
  \begin{align*}
    |f^{(l)}(\theta_1) - f^{(l)}(\theta_2)|\leq C_0|\theta_1 - \theta_2|^{\alpha - l}
  \end{align*}
for some positive constant $C_0$. Let $K(\cdot)$ be a $k$-th order kernel function satisfying that $\int K(\theta)dx = 1$, $ \int x^tK(\theta)dx = 0$ for $1\leq t\leq k-1$, and $\int |x^kK(\theta)|dx<\infty$. Defining the Weierstrass transform as 
\begin{align*}
  W_h^{(K)} f(\theta) = \int h^{-1}K\bigg(\frac{\theta-t}{h}\bigg)f(t)dt = \int K_h(\theta-t)f(t)dt,
\end{align*}
we have
  \begin{align*}
   \max_{\theta\in \mathcal{R}}|W_h^{(K)} f(\theta) - f(\theta)|\leq \frac{M_\gamma \kappa_\gamma(K)}{\lfloor \gamma \rfloor !}h^\gamma,
  \end{align*}
where $\gamma = \min\{\alpha, k\}$, $M_k = \max_{\theta\in \mathcal{R}} |f^{(k)}(\theta)|$, $M_\alpha = C_0$ and $\kappa_\gamma(K) = \int |t^\gamma K(t)|dt$.
\end{lemma}

\begin{proof}[\textbf{Proof of Lemma \ref{lemma:1}}]
The proof relies on the higher-order Taylor expansion on the function. If $\gamma = \alpha$, we have
\begin{align*}
  f(\theta+th) = f(\theta) + f'(\theta)th+\cdots+\frac{f^{(l-1)}(\theta)}{(l-1)!}t^{l-1}h^{l-1}+\frac{f^{(l)}(\tilde \theta)}{l!}t^{l}h^{l},
\end{align*}
where $l = \lfloor \alpha \rfloor$ and $\tilde \theta$ lies between $\theta$ and $\theta+th$. Because $l<k$, we thus have
\begin{align*}
 \int K(t)f(\theta+th)dt = \int K(t)f(\theta)dt + \int K(t)t^lh^l \frac{f^{(l)}(\tilde \theta)}{l!} = f(\theta)+ \int K(t)t^lh^l \frac{f^{(l)}(\tilde \theta)}{l!}
\end{align*}
and because $\int K(t)t^{l} = 0$,
\begin{align*}
  \bigg|\int K(t)t^lh^l& \frac{f^{(l)}(\tilde \theta)}{l!}\bigg| = \bigg|\int K(t)t^lh^l \frac{f^{(l)}(\tilde \theta)}{l!} - \int K(t)t^lh^l \frac{f^{(l)}(\theta)}{l!}\bigg| \\
&\leq l!^{-1}\int h^l|K(t)t^l||f^{(l)}(\tilde \theta) - f^{(l)}(\theta)|dt\leq C_0 l!^{-1}\int |t^\alpha K(t)| h^\alpha.
\end{align*}
Therefore,
\begin{align*}
  |W_h^{(K)} f(\theta) - f(\theta)|& = \bigg|\int K_h(\theta-t)f(t)dt - f(\theta)\bigg| = \bigg|\int K(t)f(\theta+th)dt - f(\theta)\bigg|\\
& \leq  C_0 l!^{-1}\int |t^\alpha K(t)| h^\alpha =  M_\gamma \lfloor \gamma \rfloor!^{-1}\int |t^\gamma K(t)| h^\gamma.
\end{align*}
The case $\gamma = k$ follows the same argument. Just notice that the Taylor expansion now becomes,
\begin{align*}
   f(\theta+th) = f(\theta) + f'(\theta)th+\cdots+\frac{f^{(k-1)}(\theta)}{(k-1)!}t^{k-1}h^{k-1}+\frac{f^{(k)}(\tilde \theta)}{k!}t^{k}h^{k},
\end{align*}
which entails that
\begin{align*}
  |W_h^{(K)}f(\theta) - f(\theta)|\leq M_k k!^{-1} \int |t|^k K(t) h^k = M_\gamma \lfloor \gamma \rfloor!^{-1} \int |t^\gamma K(t)| h^\gamma,
\end{align*}
where $M_k = \max_{\theta\in R} |f^{(k)}(\theta)|$, 
and thus completes the proof.
\end{proof}

\par
In this article, we only focus on the case when $K$ is chosen to be a density function (second order kernel function), and thus $\gamma\leq 2$. The result stated in Lemma \ref{lemma:1} can be naturally generalized to the multivariate case.

\begin{lemma}\label{lemma:2}
  Let $f$ be a real-valued function on $\mathcal{R}^p$. Define the multivariate Weierstrass transform as
  \begin{align*}
    W_h^{(\textbf{K})}f(\theta_1,\cdots,\theta_p) = \int f(t_1,\cdots,t_p)\prod_{j=1}^p h_j^{-1}K_j\bigg(\frac{\theta_j-t_j}{h_j}\bigg)dt_j.
  \end{align*}
If $f$ is H\"older $\alpha$ smooth with a constant $C_0$, i.e., for $l = \lfloor \alpha \rfloor$, $f$ is $l$-th differentiable and for all the $l$-th derivatives of $f$, we have
\begin{align*}
  |f^{(l)}_{i_1i_2\cdots i_l}(\theta_1,\cdots,\theta_p) - f^{(l)}_{i_1i_2\cdots i_l}(\theta'_1,\cdots,x'_p)|\leq \sum_{j=1}^p C_0|\theta_j - x'_j|^{\alpha-l} \qquad \forall  x,x'\in \mathcal{R}^p.
\end{align*}
Assuming $K_i$s are all $k$-th order kernels,  the approximation error of the Weierstrass transform follows
\begin{align*}
 \max_{\theta\in\mathcal{R}^p} |W_h^{(\textbf{K})}f(\theta) - f(\theta)|\leq \frac{M_\gamma}{\lfloor \gamma \rfloor!}\sum_{j=1}^p \kappa_\gamma(K_j) h_j^\gamma,
\end{align*}
where $\gamma$, $M_\gamma$ and $\kappa_\gamma(\cdot)$ are defined in Lemma \ref{lemma:1}.
\end{lemma}
The proof of Lemma \ref{lemma:2} is essentially an application of Lemma \ref{lemma:1} and is omitted. With Lemma 1 we proceed to prove Theorem 1. The original Theorem 1 is stated in terms of second-order differentiable function and second-order kernels. Here, we provide a more general version of Theorem 1.

{\em \textbf{Theorem 1.} If the posterior densities and the kernel functions satisfy the condition in Lemma 1 with $\alpha$ and $k$, then the distribution defined in \eqref{eq:approx} is proper and there exists a positive constant $c_0$ such that when $h^\gamma = \sum_{i=1}^m h_i^\gamma\leq c_0^{-1}$ for $\gamma = \min\{\alpha, k\}$, the total variation distance between the posterior distribution and the approximation follows
  \begin{align*}
    \|f - \tilde f\| = \bigg\| C^{-1}\prod_{i=1}^m f_i(\theta) -C_W^{-1} \prod_{i=1}^m W_{h_i}^{(K)} f_i(\theta) \bigg\|\leq 2r_0r_1^{-\gamma}h^\gamma,
  \end{align*}
where $C$ and $C_W$ are the normalizing constants, and $r_0, r_1$ are defined as
\begin{align*}
  r_0  = C^{-1}M\max_{i\in\{1,2,\cdots,m\}}\int \prod_{j\neq i}^{m} f_j(\theta)dx \qquad  r_1^\gamma = \frac{2M}{M_2 \kappa_\gamma(K)}.
\end{align*}
}

\begin{proof}[\textbf{Proof of Theorem \ref{thm:1}}]
  We merge $r_1$ into $h_i$ to simplify the notation. With this modification the result can be expressed as,
  \begin{align*}
    \|f - \tilde f\| = \frac{1}{2}\bigg\| C^{-1}\prod_{i=1}^m f_i(\theta) -C_W^{-1} \prod_{i=1}^m W_{r_1h_i}^{(K)} f_i(\theta) \bigg\|_{L_1}\leq 2r_0h^\gamma,
  \end{align*}
 The derivation is divided into two steps. In the first step, we obtain an estimate of the difference between the two products $\|\prod f_i(\theta) - \prod W_{r_1h_i}^{(K)}f_i(\theta)\|$, and then apply it in the second step to bound the total variation distance.
\par
A typical way to quantify the difference between two products is to decompose it into sums of relative differences, i.e.,
\begin{align*}
  \int \bigg|\prod_{i=1}^m &f_i(\theta) - \prod_{i=1}^m W_{r_1h_i}^{(K)} f_i(\theta)\bigg|dx \leq \sum_{i=1}^m \int |f_i(\theta) - W_{r_1h_i}^{(K)} f_i(\theta)| \prod_{j=1}^{i-1} f_j(\theta) \prod_{j=i+1}^m W_{r_1h_i}^{(K)}f_j(\theta) dx.
\end{align*}
Define $C_k = \max_{ I\subset \{1,2,\cdots,m\},|I| = k} \int \prod_{j\in I} f_j(\theta) dx$ and $\delta_k = \max_{I'\subset I\subset\{1,2,\cdots,m\},|I| = k} \|\prod_{j\in I}f_j(\theta) - \prod_{j\in I'} W_{r_1h_j}^{(K)}f_j(\theta)\prod_{j\in I\backslash I'} f_j(\theta)\|_{L_1}$, and notice that 
\begin{align*}
  \int \prod_{j=1}^{i-1} f_j(\theta)\prod_{j=i+1}^m W_{r_1h_j}^{(K)} f_j(\theta) dx& \leq \int \prod_{j\neq i} f_j(\theta) +  \bigg\| \prod_{j=1}^{i-1} f_j(\theta)\prod_{j=i+1}^m W_{r_1h_j}^{(K)} f_j(\theta) dx - \prod_{j\neq i} f_j(\theta)\bigg\|_{L_1}\\
&\leq  C_{m-1} + \delta_{m-1}.
\end{align*}
we can then bound each relative difference as
\begin{align*}
  \int |&f_i(\theta) - W_{r_1h_i}^{(K)} f_i(\theta)| \prod_{j=1}^{i-1} f_j(\theta) \prod_{j=i+1}^m W_{r_1h_i}^{(K)}f_j(\theta) dx\\
&\leq \max_{\theta\in\mathcal{R}} |f_i(\theta) - W_{r_1h_i}^{(K)}f_i(\theta)|\int \prod_{j=1}^{i-1} f_j(\theta) \prod_{j=i+1}^m W_{r_1h_i}^{(K)} f_j(\theta) dx\\
&\leq \frac{M_\gamma\kappa_\gamma(K)}{\lfloor \gamma \rfloor !}r_1^\gamma h_i^\gamma \cdot \int \prod_{j=1}^{i-1} f_j(\theta)\prod_{j=i+1}^m W_{r_1h_j}^{(K)} f_j(\theta) dx\\
&\leq M(C_{m-1}+\delta_{m-1})h_i^\gamma.
\end{align*}
Summing over all $i\in\{1,2,\cdot,m\}$, we have
\begin{align*}
 \int \bigg|\prod_{i=1}^m &f_i(\theta) - \prod_{i=1}^m W_{r_1h_i}^{(K)} f_i(\theta)\bigg|dx \leq MC_{m-1}h^\gamma+M\delta_{m-1}h^\gamma,
\end{align*}
where $h^\gamma = \sum_{i=1}^m h_i^\gamma$. Using the same trick to decompose $\delta_k$ entails that (noticing that $\sum_{i\in I\subset\{1,2,\cdots,m\}} h_i^\gamma\leq h^\gamma$),
\begin{align*}
  \delta_k\leq MC_{k-1}h^\gamma+ M\delta_{k-1}h^\gamma \qquad \forall k\in\{3,4,\cdots,m-1\}
\end{align*}
and for $\delta_2$, a direct computation shows that $\delta_2\leq M$ if all $h_i^\gamma\leq 1/2$.
 Defining $c_0 = \max\{\frac{M^2}{2C_3}, \frac{MC_k}{2C_{k+1}}, 2\leq k\leq m-1\}$, by mathematical induction, it is easy to verify that for $h^\gamma\leq c_0^{-1}$ and $h_i^\gamma\leq 1/2$,
\begin{align*}
  \delta_k\leq C_k \qquad\forall k\in \{3,\cdots, m-1\}
\end{align*}
and therefore we have
\begin{align}
  \int \bigg|\prod_{i=1}^mf_i(\theta) - \prod_{i=1}^m W_{r_1h_i}^{(K)} f_i(\theta)\bigg|dx\leq 2MC_{m-1}h = 2Cr_0h^\gamma \label{eq:thm1.2}.
\end{align}
An application of \eqref{eq:thm1.2} can help deriving the difference between the two normalizing constants,
\begin{align*}
  |C - C_W| \leq \int \bigg|\prod_{i=1}^m f_i(\theta) - \prod_{i=1}^m W_{r_1h_i}^{(K)} f_i(\theta)\bigg|dx \leq  2Cr_0h^\gamma.
\end{align*}
We then bound the $L_1$ distance. For $h_i^\gamma\leq 1/2$ and $h^\gamma\leq c_0^{-1}$ we have
\begin{align*}
  \int &\bigg|C^{-1}\prod_{i=1}^m f_i(\theta) -C_W^{-1} \prod_{i=1}^m W_{r_1h_i}^{(K)} f_i(\theta)\bigg|dx\\
&\leq C^{-1}\int \bigg|\prod_{i=1}^m f_i(\theta) - \prod_{i=1}^m W_{r_1h_i}^{(K)} f_i(\theta)\bigg|dx + |C^{-1} - C_W^{-1}|\int \prod_{i=1}^m W_{r_1h_i}^{(K)} f_i(\theta)dx\\
&\leq 2r_0h^\gamma + \frac{|C - C_W|}{C_WC}\cdot C_W\\
&\leq 4r_0h^\gamma,
\end{align*}
which gives the error bound stated in the theorem.
\end{proof}

\par
\textbf{\em Remarks:} Theorem 1 is a special case of the above theorem with $\alpha \ge 2$ and $k = 2$. Though the features of Weierstrass sampler will not rely on asymptotics (no requirement on the sample size), a rough asymptotic analysis on $r_0$ and $r_1$ can provide a general idea on their magnitudes and behavior with the change of sample size.
\par
If the likelihood function satisfies certain regularity conditions, local asymptotic normality will ensure the posterior density converges to a normal density both pointwise and in $L_1$. Replacing the posterior by its asymptotic distribution, i.e., substituting $f_i(\theta)$ with $N(\hat \theta_i, \Sigma_{n/m}) = N(\hat \theta_i, mI^{-1}/n)$, where $I$ is the Fisher's information matrix, and $\hat \theta_i$s are locally consistent estimators which satisfy that $\hat \theta_i - \theta_{true} \sim  N(0, \Sigma_{n/m})$ asymptotically,  we have that
\begin{align*}
  C_k&\approx \int \prod_{i=1}^k dN(\hat \theta_i, mI^{-1}/n)\\
 & = (2\pi)^{-(k-1)p/2}|\Sigma_{n/m}|^{-(k-1)/2}k^{-p/2}\exp\bigg[ -\frac{k}{2}\big\{ k^{-1}\sum_{i=1}^k( \hat \theta_i - \bar \theta)^T\Sigma_{n/m}^{-1}( \hat \theta_i - \bar \theta)\big\}\bigg],
\end{align*}
where $\bar \theta = k^{-1}\sum_{i=1}^k \hat \theta_i$. Since $\hat \theta_i$ follows $N(\theta_{true}, \Sigma_{n/m})$ roughly, we can replace $k^{-1}\sum_{i=1}^k( \hat \theta_i - \bar \theta)^T\Sigma_{n/m}^{-1}( \hat \theta_i - \bar \theta)$ by $p$ for large values of $k$ and obtain a more concise approximation as
\begin{align*}
  C_k\approx (2\pi)^{-(k-1)p/2}|\Sigma_{n/m}|^{-(k-1)/2}k^{-p/2}\exp\bigg(-\frac{pk}{2}\bigg).
\end{align*}
The above approximation suggests that $MC_{k}/C_{k+1}\approx (e+e/k)^{p/2}$, where $M\approx (2\pi)^{-p/2}|\Sigma_{n/m}|^{-1/2}$.
Therefore, the following approximations hold asymptotically
\begin{align*}
 r_0 \approx \bigg(e+\frac{e}{m-1}\bigg)^{p/2} \quad\mbox{and}\quad c_0^{-1}\geq 2(2e)^{-p/2}.
\end{align*}
To quantify $r_1$ requires more elaborate analysis, because the result depends on the choice of kernel functions. If the kernel function is a density function with $\gamma =  2$ (sufficiently smooth likelihood), $r_1 = O(\sqrt{M/M_2})$.

\par
\begin{proof}[\textbf{Proof of Theorem 3}]
 By definition we have,
  \begin{align*}
    \int_{-\infty}^x g_i(t)dt& = P\bigg\{ \theta_i\leq x~|~ c^{-m+1}\prod_{k\neq i}^m K \bigg(\frac{\theta_k-\theta_i}{h_k}\bigg) \geq u\bigg\}\\
&= C'\int_{-\infty}^x f_i(t_i) \int \prod_{k\neq i}^m \frac{1}{r_1h_k} K (\frac{t_k-t_i}{h_k}) f_k(t_k) \prod_{k=1}^m dt_k,
  \end{align*}
where $C'= (r_1/c)^{m-1}\prod_{k\neq i}h_k$. As a result, we have,
\begin{align*}
  g_i(\theta) = C'f_i(\theta)\prod_{k\neq i}^m \int \frac{1}{h_k} K \bigg(\frac{t_k-x}{r_1h_k}\bigg) f_k(t_k) dt_k = C'f_i(\theta)\prod_{k\neq i} W_{h_k}^{(K)}f_k(\theta).
\end{align*}
The right hand side is essentially the same as the expression $\prod_{i=1}^m W_{h_i}^{(K)} f_i(\theta)$ in Theorem 1, except for the $i^{th}$ term. Therefore, following the same argument in Theorem 1, we can prove that
\begin{align*}
 \bigg\|g_i(\theta) - C^{-1}\prod_{k=1}^m f_k(\theta)\bigg\|_{L_1}\leq 4r_0r_1^{-2}\sum_{k\neq i} h_k^2.
\end{align*}
\end{proof}

\subsection*{Proof of Theorem 2}
To prove Theorem 2 we first need three lemmas with proofs provided immediately after the statement of lemma. It is worth noting that the Lemma \ref{lemma:2.2} is a common tool used in showing geometric ergodicity of MCMC algorithms \citep{Johnson:2009}. The original proof used the common coupling inequality trick, while here we will adopt a different approach to derive a slightly different conclusion that is more useful for this paper.
\begin{lemma}\label{lemma:2.1}
  Assume $f(x)$ and $f_0(x)$ are two continuous density functions. For any $p_0\in (0, 1)$, there always exists a bounded measurable set $D$ such that
  \begin{align*}
    \int_D f(x)dx = \int_D f_0(x)dx \quad\mbox{and}\quad \int_D f(x)dx >p_0.
  \end{align*}

  \begin{proof}[\textbf{Proof of Lemma \ref{lemma:2.1}}]
    Let $A^+ =\{x: f(x)>f_0(x)\}, A^- = \{x: f(x)<f_0(x)\}, A = \{x: f(x) = f_0(x)\}$. Define the following functions
    \begin{align*}
      H^+(t) = \int_{A^+\cap \{|x|\leq t\}} f(x),\quad H^-(t) = \int_{A^-\cap \{|x|\leq t\}} f(x),\quad H(t) = \int_{A\cap \{|x|\leq t\}} f(x)
    \end{align*}
and
\begin{align*}
  G^+(t) = \int_{A^+\cap\{|x|\leq t\}} f(x) - f_0(x), \qquad G^-(t) = \int_{A^-\cap\{|x|\leq t\}} f_0(x) - f(x).
\end{align*}
Because $f,f_0$ are continuous, the above functions are all smooth and non-decreasing functions for $t\geq 0$, and satisfy that
\begin{align*}
  H^+(\infty) + H^-(\infty) + H(\infty) = 1,\qquad G^+(\infty) = G^-(\infty) = \|f - f_0\|.
\end{align*}
From the assumption that $\|f - f_0\|>0$ (otherwise the result is trivial), we are guaranteed that $H^+(\infty)$ is positive. The following proof is divided into two parts.
\par
{\em Case 1:} If $H^-(\infty)>0$.
Define $\epsilon = (1 - p_0)/3$. Due to the continuity and monotonic property of $H^+(\cdot), H^-(\cdot)$ and $H(\cdot)$, we can always find $t_1$ and $t_2$ such that
\begin{align*}
  H^+(\infty) - \epsilon < H^+(t_1) < H^+(\infty), \qquad  H^-(\infty) - \epsilon < H^-(t_2) < H^-(\infty) 
\end{align*}
and $t_3$ such that
\begin{align*}
H(\infty) - \epsilon < H(t_3) \leq H(\infty).
\end{align*}
\par
 Now if $G^+(t_1) = G^-(t_2)$, we can simply set $D = A^+\cap\{|x|\leq t_1\}\bigcup A^-\cap\{|x|\leq t_2\}\bigcup A\cap\{|x|\leq t_3\}$. Otherwise, without loss of generality, assuming $G^+(t_1) > G^-(t_2)$, due to the choice of $t_1$, we know that $H^+(t_1)< H^+(\infty)$, which ensures
 \begin{align*}
   G^-(t_2) < G^+(t_1)< G^+(\infty) = G^-(\infty).
 \end{align*}
Again making use of the continuity and the property of limit, we are able to find $t_2'>t_2$ such that
\begin{align*}
  G^+(t_1) = G^-(t_2')
\end{align*}
and $D$ is then taken to be  $D = A^+\cap\{|x|\leq t_1\}\bigcup A^-\cap\{|x|\leq t_2'\}\bigcup A\cap\{|x|\leq t_3\}$
\par
{\em Case 2:} If $H^-(\infty) = 0$. This case is even simpler. Define $\epsilon = (1- p_0)/2$. Since $H^+(\infty)>0$, there exist $t_1$ and $t_3$ such that,
\begin{align*}
  H^+(\infty) - \epsilon < H^+(t_1) < H^+(\infty),\qquad H(\infty) - \epsilon < H(t_3) \leq H(\infty),
\end{align*}
which at same time guarantees that $G^+(t_1)< G^+(\infty)$. Therefore, by a similar argument to Case 1, we will be able to find $t_2$ such that
\begin{align*}
  G^+(t_1) = G^-(t_2).
\end{align*}
Then $D =  A^+\cap\{|x|\leq t_1\}\bigcup A^-\cap\{|x|\leq t_2\}\bigcup A\cap\{|x|\leq t_3\}$. It is easy to verify that the set $D$ defined in the proof satisfies the properties in the lemma.
  \end{proof}

\end{lemma}

\begin{lemma}\label{lemma:2.2}
  Assume the Gibbs sampler defines a transition kernel $\kappa(\cdot, x)$,
  \begin{align*}
    P(A,x) = P(X_t\in A|X_{t-1} = x) = \int_A \kappa(t,x)dt.
  \end{align*}
 Let $f(x)$ denote the equilibrium distribution and $f_0(x)$ the approximation to $f(x)$. For any measurable set $D$ which satisfies that 
\begin{align*}
  \int_D f_0(x)dx = \int_D f(x)dx,
\end{align*}
if there exists a probability density $q$ and a positive value $\epsilon$ such that
  \begin{align*}
    \kappa(t,x) \geq \epsilon q(t)
  \end{align*}
for any $t\in \mathcal{R}$ and $x \in D$, then we have
\begin{align*}
  \int |f_1(x) - f(x)|dx \leq (1-\epsilon)\int_D |f_0(x) - f(x)|dx + \int_{D^c} |f_0(x) - f(x)|dx,
\end{align*}
where $f_1(x) = \int \kappa(x,t)f_0(t)dt$. If $D = \mathcal{R}$, the conclusion becomes,
\begin{align*}
  \|f_1(x) - f(x)\| \leq (1 - \epsilon) \|f_0(x) - f(x)\|,
\end{align*}
where $\|\cdot\|$ denots the total variation distance.
\end{lemma}

\begin{proof}[\textbf{Proof of Lemma \ref{lemma:2.2}}]
  Because  $f(x)$ is the equilibrium distribution of the Gibbs sampler, we have
  \begin{align*}
    f(x) = \int \kappa(x,t)f(t)dt.
  \end{align*}
Therefore, 
\begin{align*}
\int |f_1(x) - f(x)| &= \int \bigg|\int \kappa(x,t)f_0(t)dt - \int \kappa(x,t)f(t)dt\bigg|dx\\
 &=\int \bigg|\int \kappa(x,t)\bigg\{f_0(t) - f(t)\bigg\}dt\bigg|dx\\
 &=\int \bigg|\int_D \kappa(x,t)\bigg\{f_0(t) - f(t)\bigg\}dt\bigg|dx + \int\bigg|\int_{D^c} \kappa(x,t)\bigg\{f_0(t) - f(t)\bigg\}dt \bigg|dx.
\end{align*}
For the second term we have
\begin{align*}
  \int\bigg|\int_{D^c} \kappa(x,t)\bigg\{f_0(t) - f(t)\bigg\}dt \bigg|dx \leq \int \int_{D^c} \kappa(x,t) |f_0(t) - f(t)|dtdx = \int_{D^c} |f_0(t) - f(t)|dt,
\end{align*}
where the last equality is due to Fubini's Theorem. For the first term, notice that $\kappa(x,t) = \epsilon q(x) + (1-\epsilon)\frac{\kappa(x,t) -\epsilon q(x)}{1-\epsilon}$, where 
\begin{align*}
\frac{\kappa(x,t) -\epsilon q(x)}{1-\epsilon}>0\mbox{ and }
\int \frac{\kappa(x,t) -\epsilon q(x)}{1-\epsilon} dx = 1
\end{align*}
for $t\in D$. As a result, we have
\begin{align*}
\int \bigg|\int_D& \kappa(x,t)\bigg\{f_0(t) - f(t)\bigg\}dt\bigg|dx \\
&= \int \bigg|\epsilon\int_D q(x)\bigg\{f_0(t) - f(t)\bigg\}dt + (1-\epsilon)\int_D\frac{\kappa(x,t) - \epsilon q(x)}{1 - \epsilon}\bigg\{f_0(t) - f(t)\bigg\}dt\bigg|dx\\
&= \int \bigg|(1-\epsilon)\int_D\frac{\kappa(x,t) - \epsilon q(x)}{1 - \epsilon}\bigg\{f_0(t) - f(t)\bigg\}dt\bigg|dx\\
&\leq (1-\epsilon) \int\int_D\frac{\kappa(x,t) - \epsilon q(x)}{1 - \epsilon}|f_0(t) - f(t)|dtdx\\
&= (1 - \epsilon)\int_D|f_0(t) - f(t)|dt.
\end{align*}
Consequently, we have,
\begin{align*}
  \int |f_1(x) - f(x)| \leq (1 - \epsilon)\int_D|f_0(x) - f(x)|dx + \int_{D^c}|f_0(x) - f(x)|dx,
\end{align*}
which completes the proof.
\end{proof}

\begin{lemma}\label{lemma:2.3}
  Considering the following Gibbs sampler,
  \begin{align*}
    \theta|t_i &\sim \prod_{i = 1}^m K_h(\theta - t_i)\\
    t_i|\theta &\sim K_h(\theta - t_i)f_i(t_i),
  \end{align*}
which defines a transition kernel on $\theta$ as $\kappa(\cdot,\theta)$. If the kernel $K$ is fully supported on $\mathcal{R}$, then for any bounded measurable set $D$, there exists an $\epsilon>0$ and a probability density $q(\theta)$ such that
\begin{align*}
  \kappa(\theta, \theta_0) > \epsilon q(\theta)
\end{align*}
for any $\theta_{0} \in D$. Furthermore, if the condition \eqref{eq:thm2} is satisfied, i.e.,
\begin{align*}
  \lim_{\theta\rightarrow\infty} \inf\frac{K(\theta - t)}{W_h^{(K)}f_i(\theta)} > 0
\end{align*}
for any $t\in\mathcal{R}$ and $i\in\{1,2,\cdots,m\}$, then the set $D$ can be taken as $\mathcal{R}$.
\end{lemma}

\begin{proof}[\textbf{Proof of Lemma \ref{lemma:2.3}}]
  Consider the full Gibbs transition kernel $g(\theta,t_i, i=1,2,\cdots,m|\theta_0)$,
  \begin{align*}
    g(\theta, t|\theta_0) &= \frac{\prod_{i=1}^m K_h(\theta - t_i)}{\int\prod_{i=1}^m K_h(\theta - t_i)d\theta}\cdot\prod_{i=1}^m \frac{K(\theta_0 - t_i)f_i(t_i)}{\int K(\theta_0 - s)f_i(s)ds}\\
&=\frac{\prod_{i=1}^m K_h(\theta - t_i)}{\int\prod_{i=1}^m K_h(\theta - t_i)d\theta}\prod_{i=1}^m f_i(t_i)\cdot\prod_{i=1}^m \frac{K(\theta_0 - t_i)}{\int K(\theta_0 - s)f_i(s)ds}\\
& = q(\theta, t) \cdot \prod_{i=1}^m \frac{K(\theta_0 - t_i)}{\int K(\theta_0 - s)f_i(s)ds},
  \end{align*}
where $q(\theta, t)$ is a probability density over $\theta, t_1,\cdots,t_m$, and define
\begin{align*}
  w_D(t) = w_D(t_1,\cdots, t_m) = \min_{\theta_0\in D}\prod_{i=1}^m \frac{K(\theta_0 - t_i)}{\int K(\theta_0 - s)f_i(s)ds},
\end{align*}
which is strictly greater than 0 for any given $t_1,t_2,\cdots,t_m$ because $D$ is bounded and $K$ is strictly positive on $\mathcal{R}$. Consequently,
\begin{align*}
  g(\theta, t|\theta_0) \geq q(\theta, t) w_D(t).
\end{align*}
Let $\epsilon = \int\int q(\theta, t)w_D(t)dtd\theta$ and $q(\theta) = \epsilon^{-1}\int q(\theta, t)w_D(t) dt$. Because $w_D(t)$ is strictly positive on $\mathcal{R}$, thus we have $\epsilon > 0$. Apparently $q(\theta)$ is a probability density satisfying that,
\begin{align*}
  \kappa(\theta, \theta_0) = \int q(\theta, t|\theta_0)dt \geq \epsilon q(\theta)
\end{align*}
for any $\theta_0\in D$.
\par
Now if the condition \eqref{eq:thm2} is also satisfied, we have that
\begin{align*}
  w_R(t) = \min_{\theta \in \mathcal{R}}\prod_{i=1}^m \frac{K(\theta_0 - t_i)}{\int K(\theta_0 - s)f_i(s)ds} = \min_{\theta\in\mathcal{R}}\prod_{i=1}^m \frac{K(\theta_0 - t_i)}{W_h^{(K)}f_i(\theta_0)} > 0
\end{align*}
for any give $t_1,t_2,\cdots,t_m$, which ensures that $D$ can be chosen as $\mathcal{R}$.
\end{proof}

Theorem \ref{thm2} is a straightforward result of the above three lemmas, of which the proof is briefly described below,
\begin{proof}[\textbf{Proof of Theorem \ref{thm2}}] With Lemma \ref{lemma:2.1} we are guaranteed the existence of a set $D$ which is bounded and satisfies that
  \begin{align*}
    \int_D f(\theta)d\theta > p_0.
  \end{align*}
Now Lemma \ref{lemma:2.1} ensures the existence of a density $q(\theta)$ such that the transition kernel defined in \eqref{eq:t} and \eqref{eq:x} (with general kernel $K$) $\kappa(\cdot, \theta)$ satisfies that
\begin{align*}
  \kappa(\theta, \theta_0) \geq \epsilon q(\theta)
\end{align*}
for any $\theta_0\in D$ or for any $\theta_0\in\mathcal{R}$ if the condition \eqref{thm2} is also satisfied. Combining these facts with Lemma 4, we have the result listed in Theorem \ref{thm2}.

\end{proof}
\renewcommand{\baselinestretch}{1}
\normalsize
\bibliographystyle{apalike}

\begin{thebibliography}{}

\bibitem[Agarwal and Duchi(2012)]{Agarwal:Duchi:2012}
Agarwal, A. and Duchi, J.C. (2012).  Distributed delayed stochastic optimization, 
{\em Decision and Control (CDC), 2012 IEEE 51st Annual Conference on, IEEE}, pp. 5451-5452.

\bibitem[Ahn et al.(2012)]{Ahn:etal:2012}
Ahn, S., Korattikara, A. and Welling, M. (2012). Bayesian posterior sampling via stochastic gradient fisher scoring.
{\em Proceedings of the 29th International Conference on Machine Learning}, pp. 1591-1598.

\bibitem[Bache and Lichman(2013)]{Bache:Lichman:2013}
Bache, K. and Lichman, M. (2013). {\em UCI Machine Learning Repository [http://archive.ics.uci.edu/ml]}. 
Irvine, CA: University of California, School of Information and Computer Science.


\bibitem[Bailer-Jones and Smith(2011)]{Bailer:Jones:2011}
Bailer-Jones, C. and Smith, K. (2011). Combining probabilities.
{\em Data Processing and Analysis Consortium (DPAS)}, GAIA-C8-TN-MPIA-CBJ-053.

\bibitem[Earl and Deem(2005)]{Earl:Deem:2005}
Earl, D. J. and Deem, M. W. (2005) Parallel tempering: Theory, applications, and new perspectives.
{\em Physical Chemistry Chemical Physics}, 7, 3910

\bibitem[Fukunaga(1972)]{Fukunaga:1972}
Fukunaga, K. (1972)  Introduction to Statistical Pattern Recognition. Electrical Science Series. Academic Press, San Diego.

\bibitem[Green(1995)]{Green:1995}
Green, P. J. (1995). Reversible jump Markov chain Monte Carlo computation and Bayesian model determination. 
{\em Biometrika}, 82, 711-732.

\bibitem[Jasra et al.(2005)]{Jasra:etal:2005}
Jasra, A., Holmes, C. C. and Stephens, D. A. (2005) Markov chain Monte Carlo methods and the label switching
problem in Bayesian mixture modelling. 
{\em Statistal Science}, 20, 50-67.

\bibitem[Johnson(2009)]{Johnson:2009}
Johnson, A. A. (2009). Geometric Ergodicity of Gibbs Samplers PhD thesis, University of Minnesota, School of Statistics.


\bibitem[Liu and Ihler(2012)]{Liu:Ihler:2012}
Liu, Q. and Ihler, A. (2012). Distributed Parameter Estimation via Pseudo-likelihood,
{\em International Conference on Machine learning (ICML)}, June 2012

\bibitem[Neal(1993)]{Neal:1993}
Neal, R. M. (1993) Probabilistic inference using Markov Chain Monte Carlo Methods.
{\em Technical Report}.
University of Toronto, Toronto. (Available from http://www.cs.toronto.edu/$\sim$radford/review.
abstract.html.)

\bibitem[Neal(2010)]{Neal:2010}
Neal, R. M. (2010). MCMC using Hamiltonian dynamics. In 
{\em Handbook of Markov Chain Monte Carlo} (eds
S. Brooks, A. Gelman, G. Jones and X.-L Meng). Boca Raton: Chapman and Hall-CRC Press.

\bibitem[Neiswanger et al.(2013)]{Neis:etal:2013}
Neiswanger, W., Wang, C. and Xing, E. (2013). Asymptotically Exact, Embarrassingly Parallel MCMC, arXiv:1311.4780.

\bibitem[Polson et al.(2013)]{Polson:etal:2013}
Polson, N. G., Scott, J. G. and Windle, J. (2013). Bayesian inference for logistic models using Polya-Gamma latent variables.
{\em Journal of the American Statistical Association}, pp. 1339-1349

\bibitem[Ripley(1987)]{Ripley:1987}
Ripley, B. D. (1987). {\em Stochastic Simulation}, Wiley \& Sons.

\bibitem[Scott et al.(2013)]{Scott:Blocker:Bonassi:2013}
Scott, S. L., Blocker, A. W., and Bonassi, F. V. (2013). Bayes and big data: The consensus Monte
Carlo algorithm.
{\em Bayes 250}.

\bibitem[Smola and Narayanamurthy(2010)]{Smola:Naray:2010}
Smola, A. and Narayanamurthy, S. (2010). An architecture for parallel topic models, 
{\em Proceedings of
the VLDB Endowment}, 3, no. 1-2, 703-710.

\bibitem[Van der Vaart(1998)]{Van:1998}
Van der Vaart A. W. (1998). {\em Asymptotic Statistics}. Cambridge: Cambridge University Press.

\bibitem[Weierstrass(1885)]{Weier:1885}
Weierstrass, K. (1885). \"{U}ber die analytische Darstellbarkeit sogenannter willkürlicher Functionen einer reellen Veränderlichen.
{\em Sitzungsberichte der K\"{o}niglich Preußischen Akademie der Wissenschaften zu Berlin}, (II). Erste Mitteilung (part 1) pp. 633–639, Zweite Mitteilung (part 2) pp. 789-805.

\bibitem[Welling and Teh(2011)]{Welling:Teh:2011}
Welling, M. and Teh, Y. (2011). Bayesian learning via stochastic gradient Langevin dynamics, {\em Proceedings of
the 28th International Conference on Machine Learning (ICML)}, pp. 681-688.

\bibitem[White et al. (2013)]{White:etal:2013}
White, S.R., Kypraios, T., and Preston, S.P. (2013). Piecewise approximate Bayesian computation: fast inference for discretely observed Markov models using a factorised posterior distribution. {\em Statistics and Computing}, 1-13.
\end{thebibliography}

\end{document}